%% file: Marginal_spectral_distributions.tex
\documentclass[11pt]{article}
\input{qudit_preamble.tex}

\begin{document}

%====================================================================================================%
\title{\Large \bf Marginal spectral distributions on regular bipartite unitary orbits}
%====================================================================================================%

\author{\blue{Lin Zhang}\footnote{E-mail: godyalin@163.com}\\
  {\it\small School of Mathematical Sciences, Hangzhou Dianzi University, Hangzhou 310018, PR~China}}
\date{}
\maketitle

\begin{abstract}
Fix the spectrum of a bipartite density matrix and randomize its
eigenbasis according to Haar measure. We study the probability
distributions induced on the spectra of the two marginal states. For
arbitrary subsystem dimensions $m$ and $n$, the joint characteristic
function of the reduced density matrices is expressed as a
Harish-Chandra-Itzykson-Zuber integral whose external eigenvalues
are the pairwise sums $x_i+y_j$. Repeated external eigenvalues are
handled by confluent determinant limits. In the two-qubit case, we
derive an explicit alternating-spline formula for the joint density
of the two marginal Bloch radii. Its support is the Bravyi-Klyachko
compatibility region. We also obtain a compact truncated-power
formula for the Bloch-radius density of either individual qubit
marginal. In the qubit-qutrit case, we derive a truncated-power
formula for the qubit Bloch-radius density and a bivariate spline
formula for the joint density of the largest and smallest
eigenvalues of the qutrit marginal. The latter two variables
determine the full qutrit spectrum because the trace is fixed. The
derivations combine confluent HCIZ integrals, distributional Fourier
inversion, orbital measures, and the $\SU(2)$ and $\SU(3)$
derivative principles. The resulting densities are piecewise
polynomial on chambers determined by subset sums of the fixed global
eigenvalues, in agreement with the Duistermaat-Heckman description
of projected coadjoint-orbit measures.
\end{abstract}

\newpage\tableofcontents\newpage

%===================================================%
\section{Introduction}
%===================================================%

The quantum marginal problem is a central compatibility problem in
quantum information theory. Given density matrices assigned to
several subsystems, one asks whether they can arise as reduced
states of a common global quantum state
\cite{Davidson1976,Walter2014}. In the bipartite setting, let
$$
\cH_A=\bbC^m,\quad\cH_B=\bbC^n,\quad N=mn,
$$
and let $\rho_{AB}$ be a density matrix on $\cH_A\ot\cH_B$. Its
marginal states are $\rho_A=\ptr{B}{\rho_{AB}}$ and
$\rho_B=\ptr{A}{\rho_{AB}}$. When the spectrum of $\rho_{AB}$ is
fixed, the deterministic quantum marginal problem asks which pairs
of local spectra $(\Spec(\rho_A),\Spec(\rho_B))$ are compatible with
that global spectrum. General solutions can be formulated in terms
of representation-theoretic inequalities, moment polytopes and
symplectic reduction. Klyachko's work \cite{Klyachko2004} provides a
general representation-theoretic framework, while explicit
low-dimensional descriptions include Bravyi's inequalities for
two-qubit states \cite{Bravyi2004}. From the symplectic viewpoint
\cite{Silva2008}, the compatible local spectra form the Kirwan
polytope \cite{Kirwan1984} associated with the action of the local
unitary group on a global coadjoint orbit.

The deterministic compatibility region does not, however, reveal how
marginal spectra are distributed inside that region. A natural
probabilistic refinement is therefore to fix the global eigenvalues
and randomize only the global eigenbasis. This model separates the
effect of the global spectrum from that of Haar-distributed
eigenvectors.

Let
$$
\bdlambda=(\lambda_1,\ldots,\lambda_N),\quad\lambda_1>\lambda_2>\cdots>\lambda_N\geqslant0,\quad
\sum^N_{j=1}\lambda_j=1,
$$
and define
$\Lambda=\diag(\lambda_1,\ldots,\lambda_N)=\diag(\bdlambda)$. The
associated regular unitary orbit is
$$
\cU_\Lambda:=\Set{\bsU\Lambda\bsU^\dagger: \bsU\in\sfU(N)}.
$$
We equip $\cU_\Lambda$ with the orbital probability measure
\cite{Olshanski2013} obtained by pushing normalized Haar measure on
$\sfU(N)$ forward under $\bsU\mapsto \bsU\Lambda\bsU^\dagger$. Thus,
$\rho_{AB}=\bsU\Lambda\bsU^\dagger$, where $\bsU\sim
(\sfU(N),\mu_\haar)$ is a random global state with fixed spectrum
$\bdlambda$.

The object of interest is the push-forward of this orbital measure
under the marginal map \cite{Collins2023,Walter2014}
$$
\Phi: \cU_\Lambda\to \rD(\bbC^m)\times\rD(\bbC^n),\quad
\Phi(\rho_{AB})=(\ptr{B}{\rho_{AB}},\ptr{A}{\rho_{AB}}),
$$
where $\rD(\bbC^d)$ denotes the set of $d\times d$ density matrices.
In particular, we seek the joint distribution of
$(\rho_A,\rho_B)=(\ptr{B}{\bsU\Lambda\bsU^\dagger},\ptr{A}{\bsU\Lambda\bsU^\dagger})$,
and ultimately the joint distribution of their ordered eigenvalues.
The starting point is the characteristic function. For Hermitian
test matrices $\bsX\in\Herm(\bbC^m)$ and $\bsY\in\Herm(\bbC^n)$, the
defining property of the partial trace gives
$$
\Tr{\bsX\rho_A}+\Tr{\bsY\rho_B}=\Tr{(\bsX\ot\I_n+\I_m\ot\bsY)\rho_{AB}}.
$$
Consequently,
\begin{eqnarray*}
\varphi_{\bdlambda}(\bsX,\bsY) =
\int_{\sfU(N)}\exp\Pa{\mathrm{i}(\bsX\ot\I_n+\I_m\ot\bsY)\bsU\Lambda\bsU^\dagger}\dif\mu_{\haar}(\bsU).
\end{eqnarray*}
This is a Harish-Chandra-Itzykson-Zuber integral
\cite{Harish1975,IZ1980,McSwiggen2018}. If $x_1,\ldots,x_m$ and
$y_1,\ldots,y_n$ are the eigenvalues of $\bsX$ and $\bsY$,
respectively, then the eigenvalues of the external matrix
$\bsX\ot\I_n+\I_m\ot\bsY$ are
$$
x_i+y_j,\quad 1\leqslant i\leqslant m,\quad 1\leqslant j\leqslant n.
$$
The problem therefore reduces to a \emph{confluent HCIZ
integral}\footnote{The confluent HCIZ integral is a critical
generalization of the standard HCIZ integral to the case of
eigenvalue coalescence (i.e., degeneracy). In this singular regime,
the standard formula breaks down due to vanishing denominators,
necessitating new mathematical tools. To address this, the confluent
variant relies on limiting procedures or determinant/derivative
structures to obtain a meaningful generalization.} whenever some of
these pairwise sums coincide.

The relation with symplectic geometry is important. A regular
unitary orbit is a compact coadjoint orbit and hence a compact
symplectic manifold. Push-forwards of its Liouville measure under
moment maps are Duistermaat-Heckman measures
\cite{Bytsenko2005,Christandl2014,DH1982,Zhang2019}. Their densities
are piecewise polynomial on the chambers of a finite hyperplane
arrangement, which explains why truncated powers, box splines, and
multivariate splines \cite{Wang1994} naturally occur in the formulas
below---a short introduction to truncated powers is given in
Appendix~\ref{app:tpf}.

We obtain the following principal results.
\begin{enumerate}
\item For arbitrary $m$ and $n$, we express the joint characteristic function
of $(\rho_A,\rho_B)$ as a confluent HCIZ determinant.
\item For a two-qubit state, we derive an explicit piecewise-polynomial density for the pair of marginal Bloch radii. Its support is exactly the two-qubit compatibility polytope.
\item For one marginal of a two-qubit state, we derive a compact truncated-power formula for the Bloch-radius density and hence for the ordered marginal eigenvalues.
\item For a qubit-qutrit state, we obtain an explicit truncated-power formula for the qubit Bloch-radius density; for the qutrit marginal, we obtain a bivariate spline formula for
the joint density of the largest and smallest eigenvalues. Since the
trace is fixed, these two variables determine the complete qutrit
spectrum.
\end{enumerate}
Throughout most of the paper, we assume that the global spectrum is
regular/non-degenerate, i.e., all eigenvalues are distinct.
Degenerate spectra can be treated by continuous confluent limits in
the variables involved.

The paper proceeds as follows. Section~\ref{sect:2} covers the
technical preliminaries. Section~\ref{sect:3} derives the joint
characteristic function. Sections~\ref{sect:4} and~\ref{sect:5}
handle the two-qubit and qubit-qutrit cases, respectively.
Section~\ref{sect:6} offers concluding remarks and future
directions. The appendices contain the explicit spline computations.

%===================================================%
\section{Analytic and geometric preliminaries}
\label{sect:2}
%===================================================%

\subsection{Density matrices and regular unitary orbits}

Let $\Herm(\bbC^d)$ denote the real vector space of $d\times d$
Hermitian matrices, equipped with the Hilbert-Schmidt inner product
$\Inner{\bsA}{\bsB}=\Tr{\bsA\bsB}$. The state space is
\begin{eqnarray}\label{eq:densitymatrix}
\rD(\bbC^d)=\Set{\rho\in\Herm(\bbC^d):\rho\geqslant\zero,\Tr{\rho}=1}.
\end{eqnarray}
Let
\begin{eqnarray}\label{eq:chamberN}
C_N=\Set{\bsx\in\bbR^N: x_1>x_2>\cdots >x_N}
\end{eqnarray}
be the open Weyl chamber, and let
\begin{eqnarray}\label{eq:probsimplex}
\Delta_{N-1}=\Set{\bsx\in\bbR^N_{\geqslant0}: \sum^N_{j=1}x_j=1}
\end{eqnarray}
be the closed probability simplex. For $\bdlambda\in C_N\cap
\Delta_{N-1}, \Lambda=\diag(\bdlambda)$, the spectrum is called
\emph{regular} because its entries are pairwise distinct. The
associated unitary orbit is
\begin{eqnarray}\label{eq:uorbit}
\cU_\Lambda=\Set{\bsU\Lambda\bsU^\dagger: \bsU\in\sfU(N)}.
\end{eqnarray}
This permits the rank-deficient but regular spectrum used in
Figure~\ref{fig:jointdensity}.

Let $\mu_{\haar}$ denote the normalized Haar measure on $\sfU(N)$.
The orbital measure $\nu_{\bdlambda}$ is the push-forward of
$\mu_{\haar}$ under $\bsU\mapsto \bsU\Lambda\bsU^\dagger$
\cite{Olshanski2013}.

For $N=mn$, define the marginal map
\begin{eqnarray}
\Phi:\cU_\Lambda\to\rD(\bbC^m)\times\rD(\bbC^n),\quad
\Phi(\rho)=(\ptr{B}{\rho_{AB}},\ptr{A}{\rho_{AB}}).
\end{eqnarray}
The probability distribution studied below is
\begin{eqnarray}
\mu^{AB}_{\bdlambda}=\Phi_*\nu_{\bdlambda}.
\end{eqnarray}
Both marginals have unit trace. Accordingly, $\mu^{AB}_{\bdlambda}$
is supported on the affine space
$$
\Set{(\bsA,\bsB)\in \Herm(\bbC^m)\times
\Herm(\bbC^n):\Tr{\bsA}=\Tr{\bsB}=1}.
$$
Whenever a density is used, it is understood with respect to the
natural Lebesgue measure on this affine space, or with respect to
the corresponding eigenvalue coordinates after radialization.
Throughout the whole paper, $\chi_S$ denote the indicator of a set
$S$.

\subsection{Fourier conventions}

For an integrable function $f$ on $\bbR^d$, we use
\begin{eqnarray}\label{eq:f-transform}
\widehat
f(\xi):=\cF(f)(\xi)=\int_{\bbR^d}f(\bsx)e^{\mathrm{i}\Inner{\xi}{\bsx}}[\dif\bsx],
\end{eqnarray}
where $[\dif\bsx]:=\prod^d_{k=1}\dif x_k$ for
$\bsx=(x_1,\ldots,x_d)$ and
\begin{eqnarray}\label{eq:inv-f-transform}
f(\bsx)=\cF^{-1}(\widehat
f)(\bsx)=\frac1{(2\pi)^d}\int_{\bbR^d}\widehat
f(\xi)e^{-\mathrm{i}\Inner{\xi}{\bsx}}[\dif\xi].
\end{eqnarray}
These conventions extend to tempered distributions
\cite{Duistermaat2010}.

For a probability measure $\mu$ on a finite-dimensional real vector
space $V$, its characteristic function is
\begin{eqnarray}\label{eq:charf}
\varphi_\mu(\xi)=\int_V
e^{\mathrm{i}\Inner{\xi}{\bsx}}\dif\mu(\bsx).
\end{eqnarray}
Thus the characteristic function is the Fourier transform of the
measure, with no additional factor of $(2\pi)^{-d}$.

For a probability measure on $\Herm(\bbC^d)$, the natural pairing is
$\Inner{\bsX}{\bsH}=\Tr{\bsX\bsH}$, so that
\begin{eqnarray}\label{eq:charfonherm}
\varphi_\mu(\bsX) =
\int_{\Herm(\bbC^d)}e^{\mathrm{i}\Tr{\bsX\bsH}}\dif\mu(\bsH).
\end{eqnarray}

\subsection{The HCIZ integral}

For $\bsz=(z_1,\ldots,z_N)$, define the Vandermonde product
\begin{eqnarray}\label{eq:Vandermonde}
V_N(\bsz)=\prod_{1\leqslant i<j\leqslant N}(z_i-z_j).
\end{eqnarray}
Also set
\begin{eqnarray}\label{eq:normalization}
\gamma_N=\prod^N_{k=1}\Gamma(k).
\end{eqnarray}
The Harish-Chandra-Itzykson-Zuber formula
\cite{Harish1975,IZ1980,McSwiggen2018} reads as follows.

\begin{prop}[HCIZ formula]\label{prop:HCIZ}
Let $\bsA,\bsB\in\Herm(\bbC^N)$ have simple eigenvalue vectors
$$
\bsa=(a_1,\ldots,a_N),\quad \bsb=(b_1,\ldots,b_N).
$$
Then
\begin{eqnarray}\label{eq:HCIZforz}
\int_{\sfU(N)}e^{z\Tr{\bsA\bsU\bsB\bsU^\dagger}}\dif\mu_{\haar}(\bsU)
= \gamma_N \frac{\det\Pa{e^{z
a_ib_j}}^N_{i,j=1}}{z^{\frac{N(N-1)}2}V_N(\bsa)V_N(\bsb)}.
\end{eqnarray}
For $z=\mathrm{i}$,
\begin{eqnarray}\label{eq:HCIZfori}
\int_{\sfU(N)}e^{\mathrm{i}\Tr{\bsA\bsU\bsB\bsU^\dagger}}\dif\mu_{\haar}(\bsU)
= \gamma_N \mathrm{i}^{-\frac{N(N-1)}2}\frac{\det\Pa{e^{\mathrm{i}
a_ib_j}}^N_{i,j=1}}{V_N(\bsa)V_N(\bsb)}.
\end{eqnarray}
\end{prop}
Although the right-hand side of Eq.~\eqref{eq:HCIZfori} is displayed
using distinct eigenvalues, the HCIZ integral is an entire symmetric
function of the eigenvalues of $\bsA$ and $\bsB$. Repeated
eigenvalues are therefore handled by continuous confluent limits. A
basic confluent identity used repeatedly below is
\begin{eqnarray}\label{eq:confid}
\lim_{(t_1,\ldots,t_r)\to(t,\ldots,t)}\frac{\det\Pa{f_i(t_j)}^r_{i,j=1}}{V_r(t_1,\ldots,t_r)}
=(-1)^{\frac{r(r-1)}2}
\det\Pa{\frac{f^{(j-1)}_i(t)}{(j-1)!}}^r_{i,j=1}.
\end{eqnarray}
Here the sign is consistent with the convention
$V_r(\bst)=\prod_{1\leqslant i<j\leqslant r}(t_i-t_j)$. The same
identity applies to a confluent block of rows or columns inside a
larger determinant.

\subsection{Abelian projections and spectral distributions}

A unitary-conjugation-invariant measure on Hermitian matrices can
first be projected onto a fixed Cartan subalgebra, giving the
distribution of the diagonal entries. Its ordered eigenvalue
distribution is then recovered by applying the Weyl differential
operator and multiplying by the Weyl denominator. This is often
referred to as the \emph{derivative principle}
\cite{Christandl2014,Mejia2017,Zhang2017}.

We shall use two low-rank forms directly. For a rotationally
invariant random qubit Bloch vector
$\bsa=(a_1,a_2,a_3)^\t\in\bbR^3$, let $a=\abs{\bsa}$ and let
$\alpha=a_3$ be one fixed Cartesian component. If $p(a)$ is the
density of $a$ and $q(\alpha)$ is the density of $\alpha$, then
\begin{eqnarray}\label{eq:qvsp}
q(\alpha)=\int^1_{\abs{\alpha}}\frac{p(a)}{2a}\dif a.
\end{eqnarray}
Consequently,
\begin{eqnarray}\label{eq:pq'}
p(a)=(-2a) q'(a),\quad a>0.
\end{eqnarray}
The such relationship Eq.~\eqref{eq:pq'} between $p(a)$ and
$q(\alpha)$ is established in Appendix~\ref{app:pvsq}. For two
rotationally invariant Bloch vectors with radii $a,b$ and fixed
components $\alpha,\beta$.
\begin{eqnarray}\label{eq:qabvspab}
q(\alpha,\beta) =
\int^1_{\abs{\alpha}}\int^1_{\abs{\beta}}\frac{p(a,b)}{4ab}\dif
a\dif b,
\end{eqnarray}
and therefore
\begin{eqnarray}\label{eq:pvsppq}
p(a,b) =(4ab) \partial_a\partial_b q(a,b),\quad (a,b)\in\bbR^2_{>0}.
\end{eqnarray}
The relationship Eq.~\eqref{eq:pvsppq} between $p(a,b)$ and
$q(\alpha,\beta)$ is established in Appendix~\ref{app:pabvsqab}.

%================================================================%
\section{Joint characteristic function in arbitrary dimensions}
\label{sect:3}
%================================================================%

Let
$$
N=mn,\quad\rho_{AB}=\bsU\Lambda\bsU^\dagger,\quad\bsU\sim(\sfU(N),\mu_{\haar}).
$$
Define $\rho_A=\ptr{B}{\rho_{AB}}$ and $\rho_B=\ptr{A}{\rho_{AB}}$.
\begin{thrm}[Joint characteristic function]
For $\bsX\in\Herm(\bbC^m)$ and $\bsY\in\Herm(\bbC^n)$, the joint
characteristic function of $(\rho_A,\rho_B)$ is
\begin{eqnarray}\label{eq:jointcharf}
\varphi_{\bdlambda}(\bsX,\bsY) =
\int_{\sfU(N)}\exp\Pa{\mathrm{i}\Tr{(\bsX\ot\I_n+\I_m\ot\bsY)\bsU\Lambda\bsU^\dagger}}\dif\mu_{\haar}(\bsU).
\end{eqnarray}
It is invariant under independent unitary conjugations:
\begin{eqnarray}\label{eq:uinvariant}
\varphi_{\bdlambda}(\bsV\bsX\bsV^\dagger,\bsW\bsY\bsW^\dagger)=\varphi_{\bdlambda}(\bsX,\bsY)
\end{eqnarray}
for all $\bsV\in\sfU(m)$ and $\bsW\in\sfU(n)$. Let
$\bsx=(x_1,\ldots,x_m)$ and $\bsy=(y_1,\ldots,y_n)$ be the
eigenvalue vectors of $\bsX$ and $\bsY$, respectively. Define the
$N$-component vector
\begin{eqnarray}\label{eq:hxy}
\bsh(\bsx,\bsy)=(x_i+y_j)_{1\leqslant i\leqslant m,1\leqslant
j\leqslant n}
\end{eqnarray}
in any fixed ordering. Then
\begin{eqnarray}\label{eq:varphixy}
\varphi_{\bdlambda}(\bsx,\bsy)=\gamma_N
\mathrm{i}^{-\frac{N(N-1)}2}\lim_{\bss\to
\bsh(\bsx,\bsy)}\frac{\det\Pa{e^{\mathrm{i}s_i\lambda_j}}^N_{i,j=1}}{V_N(\bss)V_N(\bdlambda)}.
\end{eqnarray}
The limit is the continuous confluent limit at all repeated
components of $\bsh(\bsx,\bsy)$.
\end{thrm}

\begin{proof}
By the defining property of the partial trace,
\begin{eqnarray*}
&&\Tr{\bsX\rho_A} + \Tr{\bsY\rho_B} =
\Tr{(\bsX\ot\I_n)(\rho_{AB})}+\Tr{(\I_m\ot\bsY)(\rho_{AB})}\\
&&=\Tr{(\bsX\ot\I_n+\I_m\ot\bsY)\rho_{AB}}
\end{eqnarray*}
which proves Eq.~\eqref{eq:jointcharf}. For $\bsV\in\sfU(m)$ and
$\bsW\in\sfU(n)$,
$$
\bsV\bsX\bsV^\dagger\ot\I_n + \I_m\ot\bsW\bsY\bsW^\dagger
=(\bsV\ot\bsW)(\bsX\ot\I_n+\I_m\ot\bsY)(\bsV\ot\bsW)^\dagger.
$$
The invariance of Haar measure under left multiplication by
$\bsV\ot\bsW$ gives Eq.~\eqref{eq:uinvariant}. Finally, the
eigenvalues of $\bsX\ot\I_n+\I_m\ot\bsY$ are the pairwise sums
$x_i+y_j$. Eq.~\eqref{eq:varphixy} follows from the HCIZ formula and
continuous extension to repeated eigenvalues.
\end{proof}

\begin{remark}
The variables $x_i$ and $y_j$ are eigenvalues of Hermitian test
matrices and are therefore arbitrary real numbers. They are not
probability vectors and need not have unit sum.
\end{remark}

\begin{remark}[Gauge invariance]
The joint characteristic function satisfies
\begin{eqnarray}\label{eq:gaugeinvariant}
\varphi_{\bdlambda}(\bsX+c\I_m,\bsY-c\I_n)=\varphi_{\bdlambda}(\bsX,\bsY),
\end{eqnarray}
for every $c\in\bbR$. Indeed,
$$
(\bsX+c\I_m)\ot\I_n + \I_m\ot(\bsY-c\I_n) =
\bsX\ot\I_n+\I_m\ot\bsY.
$$
Thus one scalar direction in the pair of test matrices is redundant,
consistently with the two unit-trace constraints on the marginal
states.
\end{remark}

\begin{cor}[Characteristic function of one marginal]
For $\bsX\in\Herm(\bbC^m)$,
\begin{eqnarray}\label{eq:onesidecharf}
\varphi^A_{\bdlambda}(\bsX) =
\int_{\sfU(N)}e^{\mathrm{i}\Tr{(\bsX\ot\I_n)\bsU\Lambda\bsU^\dagger}}\dif\mu_{\haar}(\bsU).
\end{eqnarray}
If $x_1,\ldots,x_m$ are the eigenvalues of $\bsX$, then the external
eigenvalue vector consists of $x_i$, each repeated $n$ times:
\begin{eqnarray}\label{eq:hax}
\bsh_A(\bsx)=(x_1,\ldots,x_1,\ldots,x_m,\ldots,x_m).
\end{eqnarray}
Thus
\begin{eqnarray}\label{eq:varphiaxy}
\varphi^A_{\bdlambda}(\bsx)=\gamma_N
\mathrm{i}^{-\frac{N(N-1)}2}\lim_{\bss\to
\bsh_A(\bsx)}\frac{\det\Pa{e^{\mathrm{i}s_i\lambda_j}}^N_{i,j=1}}{V_N(\bss)V_N(\bdlambda)}.
\end{eqnarray}
An analogous formula holds for $\rho_B$.
\end{cor}

\begin{remark}
Fourier inversion of $\varphi_{\bdlambda}(\bsx,\bsy)$  in diagonal
test variables gives an Abelian, or diagonal, marginal distribution.
It does not by itself give the ordered eigenvalue density. The
latter requires radialization through the appropriate Weyl
derivative principle. This distinction is essential in
Sections~\ref{sect:4} and \ref{sect:5}.
\end{remark}

%================================================================%
\section{Two-qubit systems}
\label{sect:4}
%================================================================%

Let $m=n=2$, so that $N=4$. Every qubit state has a Bloch
representation
\begin{eqnarray}\label{eq:Blochrep}
\rho=\rho(\bsr)=\frac12(\I_2+\bsr\cdot\bdsigma),\quad\bsr\in\bbR^3,\quad
\abs{\bsr}\leqslant1,
\end{eqnarray}
where $\bdsigma=(\sigma_1,\sigma_2,\sigma_3)$ is the vector of Pauli
matrices. Write
\begin{eqnarray}\label{eq:blochbloch}
\rho_A=\frac12(\I_2+\bsa\cdot\bdsigma),\quad
\rho_B=\frac12(\I_2+\bsb\cdot\bdsigma),
\end{eqnarray}
where $\bsa=(a_1,a_2,a_3)^\t$ and $\bsb=(b_1,b_2,b_3)^\t$ with
$\abs{\bsa}=a$ and $\abs{\bsb}=b$. The ordered local spectra are
\begin{eqnarray}\label{eq:specab}
\Spec(\rho_A) = \Pa{\frac{1+a}2,\frac{1-a}2},\quad \Spec(\rho_B) =
\Pa{\frac{1+b}2,\frac{1-b}2}.
\end{eqnarray}

\subsection{Abelian characteristic function}

Define the fixed Bloch components
\begin{eqnarray}\label{eq:alphabeta}
\begin{cases}
\alpha=a_3=\Tr{\rho_A\sigma_3}=\Tr{(\sigma_3\ot\I_2)\rho_{AB}},\\
\beta=b_3=\Tr{\rho_B\sigma_3}=\Tr{(\I_2\ot\sigma_3)\rho_{AB}}.
\end{cases}
\end{eqnarray}
Let $q_{\bdlambda}(\alpha,\beta)$ denote their joint density. For
Fourier variables $s,t$, $s\alpha+t\beta=\Tr{\bsH_{s,t}\rho_{AB}}$,
where
\begin{eqnarray}\label{eq:Hst}
\bsH_{s,t}=s\sigma_3\ot\I_2+t\I_2\ot\sigma_3.
\end{eqnarray}
Its eigenvalues are
\begin{eqnarray}\label{eq:h}
\bsh=(s+t,s-t,-s+t,-s-t).
\end{eqnarray}
Their Vandermonde product is
\begin{eqnarray}\label{eq:V4h}
V_4(\bsh)=64s^2t^2(s^2-t^2).
\end{eqnarray}
For $\pi\in S_4$, the symmetric group of $\set{1,2,3,4}$, define
\begin{eqnarray}\label{eq:uv}
\begin{cases}
u_\pi(\bdlambda)= \lambda_{\pi(1)}+\lambda_{\pi(2)}-\lambda_{\pi(3)}-\lambda_{\pi(4)} = 2(\lambda_{\pi(1)}+\lambda_{\pi(2)})-1,\\
v_\pi(\bdlambda)=
\lambda_{\pi(1)}-\lambda_{\pi(2)}+\lambda_{\pi(3)}-\lambda_{\pi(4)}
= 2(\lambda_{\pi(1)}+\lambda_{\pi(3)})-1.
\end{cases}
\end{eqnarray}
Expanding the HCIZ determinant gives the following characteristic
function.

\begin{prop}\label{prop:jointcharf}
The joint characteristic function of $(\alpha,\beta)$ is
\begin{eqnarray}\label{eq:charfalphabeta}
\widehat q_{\bdlambda}(s,t) = -\frac3{16V_4(\bdlambda)}\sum_{\pi\in
S_4}\sign(\pi)\frac{e^{\mathrm{i}(su_\pi(\bdlambda)+tv_\pi(\bdlambda))}}{s^2t^2(s^2-t^2)}.
\end{eqnarray}
The apparent singularities at $s=0,t=0,s=t,s=-t$ are removable in
the complete alternating sum. Individual summands are understood
distributionally using a common polarization.
\end{prop}

\begin{proof}
For $N=4$, $\gamma_N=12$ and $\mathrm{i}^{-6}=-1$. Using
Eq.~\eqref{eq:V4h} in the HCIZ formula yields
\begin{eqnarray}
\widehat
q_{\bdlambda}(s,t)=-\frac3{16V_4(\bdlambda)}\frac{\det\Pa{e^{\mathrm{i}h_i\lambda_j}}^4_{i,j=1}}{s^2t^2(s^2-t^2)}.
\end{eqnarray}
The determinant expansion is
$$
\det\Pa{e^{\mathrm{i}h_i\lambda_j}}^4_{i,j=1}=\sum_{\pi\in
S_4}\sign(\pi)e^{\mathrm{i}\sum^4_{k=1}h_k\lambda_{\pi(k)}}.
$$
Substitution of Eq.~\eqref{eq:h} gives
$\sum^4_{k=1}h_k\lambda_{\pi(k)}=su_{\pi}(\bdlambda)+tv_{\pi}(\bdlambda)$,
which proves Eq.~\eqref{eq:charfalphabeta}.
\end{proof}

\subsection{Joint density of the Bloch radii}

Local-unitary invariance implies invariance under independent
rotations of $\bsa$ and $\bsb$. Conditional on fixed radii $(a,b)$,
the directions of the two Bloch vectors therefore have the product
of the uniform measures on $S^2\times S^2$. In particular,
$$
\alpha\sim\op{Unif}[-a,a],\quad \beta\sim\op{Unif}[-b,b],
$$
independently conditional on $(\bsa,\bsb)$. Hence
\begin{eqnarray}\label{eq:qlambdaalphabeta}
q_{\bdlambda}(\alpha,\beta)=\int^1_{\abs{\alpha}}\int^1_{\abs{\beta}}
\frac{p_{\bdlambda}(a,b)}{4ab}\dif a\dif b,
\end{eqnarray}
where $p_{\bdlambda}(a,b)$ is the joint density of the radii.
Therefore,
\begin{eqnarray}\label{eq:plambdaab}
p_{\bdlambda}(a,b)=(4ab)\partial_a\partial_bq_{\bdlambda}(a,b),\quad
(a,b)\in\bbR^2_{>0}.
\end{eqnarray}
Define
\begin{eqnarray}\label{eq:Grep}
G(x,y):=\int_{\bbR^4_{\geqslant0}}\delta\Pa{\Pa{\begin{array}{c}
                                                 x \\
                                                 y
                                               \end{array}
} - r_1\Pa{\begin{array}{c}
             1 \\
             0
           \end{array}
}-r_2\Pa{\begin{array}{c}
           0 \\
           1
         \end{array}
}-r_3\Pa{\begin{array}{c}
           1 \\
           1
         \end{array}
}-r_4\Pa{\begin{array}{c}
           1 \\
           -1
         \end{array}
}}[\dif\bsr],
\end{eqnarray}
where $[\dif\bsr]:=\dif r_1\dif r_2\dif r_3\dif r_4$ and $\delta$ is
the Dirac delta function of vector arguments \cite{Zhang2021}. This
is the \emph{multivariate truncated-power function}
\cite{Boor1993,Wang1994} associated with the four column vectors
$$
(1,0)^\t,\quad(0,1)^\t,\quad(1,1)^\t,\quad(1,-1)^\t.
$$
Its Fourier transform is the polarized distribution
$$
\widehat G(s,t) = \frac1{st(s+t)(s-t)},
$$
where the boundary values are taken with the polarization induced by
the positive-ray representation above. Eliminating $r_1,r_2$ gives
\begin{eqnarray}
G(x,y)&=&\operatorname{Area}\Set{(r_3,r_4)\in\bbR^2_{\geqslant0}:
r_3+r_4\leqslant x, r_3-r_4\leqslant y}\label{eq:Garea}\\
&=& \frac12\int^x_0(\min(\zeta,y)+\zeta)_+\dif\zeta,\quad
(x,y)\in\bbR^2. \label{eq:Gxy}
\end{eqnarray}
The explicit formula remains
\begin{eqnarray}\label{eq:Gexplicita}
G(x,y)=\begin{cases}
\frac{x^2}2,&\text{if }(x,y)\in C_1=\Set{(x,y)\mid x\geqslant 0,y\geqslant x},\\
\frac{x^2+2xy-y^2}4,&\text{if }(x,y)\in C_2=\Set{(x,y)\mid x\geqslant0,0\leqslant y<x},\\
\frac{(x+y)^2}4,&\text{if }(x,y)\in C_3=\Set{(x,y):x\geqslant0,-x< y<0},\\
0,&\text{if }(x,y)\in C_0=\bbR^2\backslash(C_1\cup C_2\cup C_3).
\end{cases}
\end{eqnarray}
Equivalently,
\begin{eqnarray}\label{eq:Gexplicitb}
G(x,y) = \frac14\Br{(y+(x)_+)^2_+ + (y-(x)_+)^2_+ - 2((y)_+)^2},
\end{eqnarray}
where $(z)_+:=\max(z,0)$. Its closed support is
$$
\op{supp}(G)=\Set{(x,y)\in\bbR^2: x\geqslant0,x+y\geqslant0}.
$$
Here the detailed derivation from Eq.~\eqref{eq:Grep} to
Eq.~\eqref{eq:Gexplicitb} is relegated to
Appendix~\ref{app:computing}. The graph of $G(x,y)$ is depicted in
the following Figure~\ref{fig:chamber-support}.
\begin{figure}[ht]
\centering \subfigure[density function]
{\begin{minipage}[b]{0.45\linewidth}
\includegraphics[width=1\textwidth]{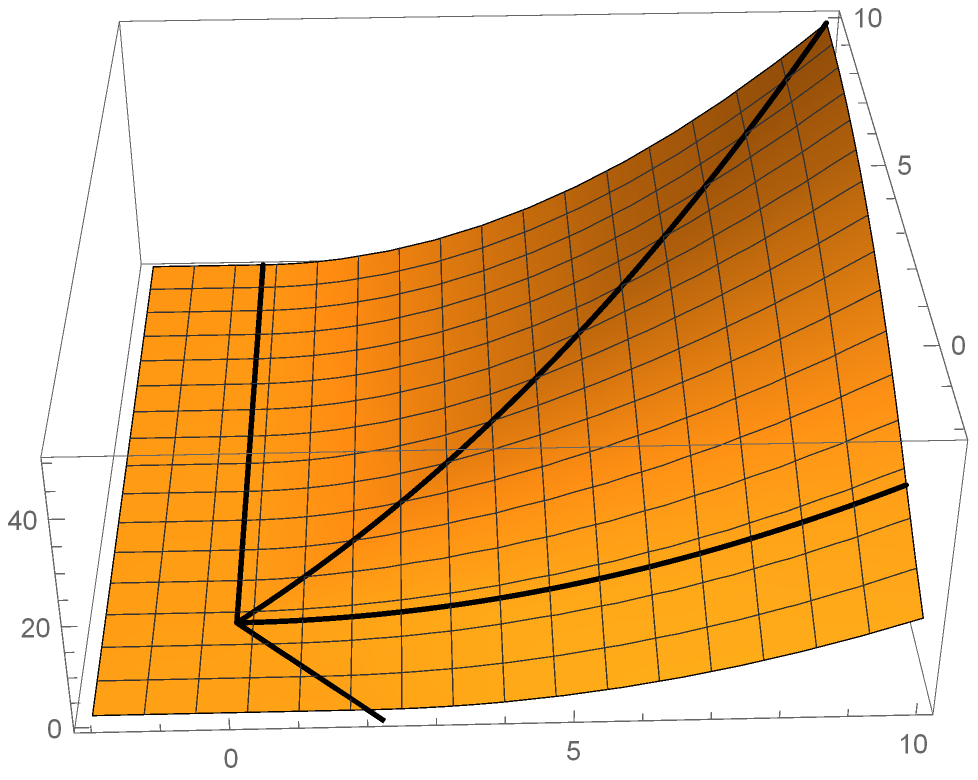}
\end{minipage}}
\centering \subfigure[support] {\begin{minipage}[b]{0.5\linewidth}
\includegraphics[width=0.7\textwidth]{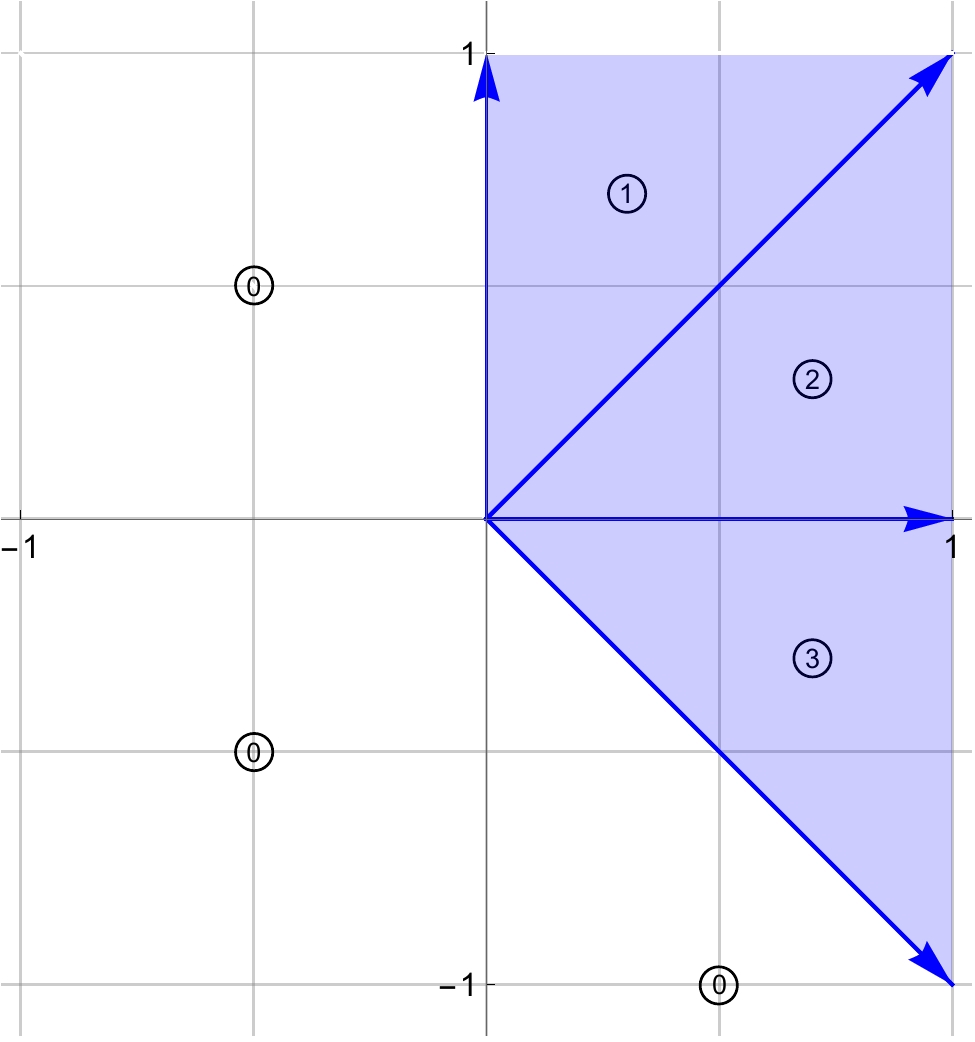}
\end{minipage}}
\caption{The graph of $G(x,y)$ and its
support}\label{fig:chamber-support}
\end{figure}

\begin{thrm}[Joint two-qubit Bloch-radius density]\label{th:2qbit}
Let $\bdlambda=(\lambda_1,\ldots,\lambda_4)\in C_4\cap\Delta_3$.
Then the joint probability density of the two marginal Bloch radii
is
\begin{eqnarray}\label{eq:jointdensityab}
p_{\bdlambda}(a,b) = \frac{3ab}{4V_4(\bdlambda)}\sum_{\pi\in
S_4}\sign(\pi)G(a-u_\pi(\bdlambda),b-v_{\pi}(\bdlambda))
\end{eqnarray}
for $(a,b)\in\bbR^2_{\geqslant0}$, with
$(u_\pi(\bdlambda),v_\pi(\bdlambda))$ given by Eq.~\eqref{eq:uv}.
The density is zero outside $[0,1]^2$.
\end{thrm}
The graph of $p_{\bdlambda}(a,b)$ is depicted in the following
Figure~\ref{fig:jointdensity}.
\begin{figure}[h!]\centering
{\begin{minipage}[b]{0.6\linewidth}
\includegraphics[width=1\textwidth]{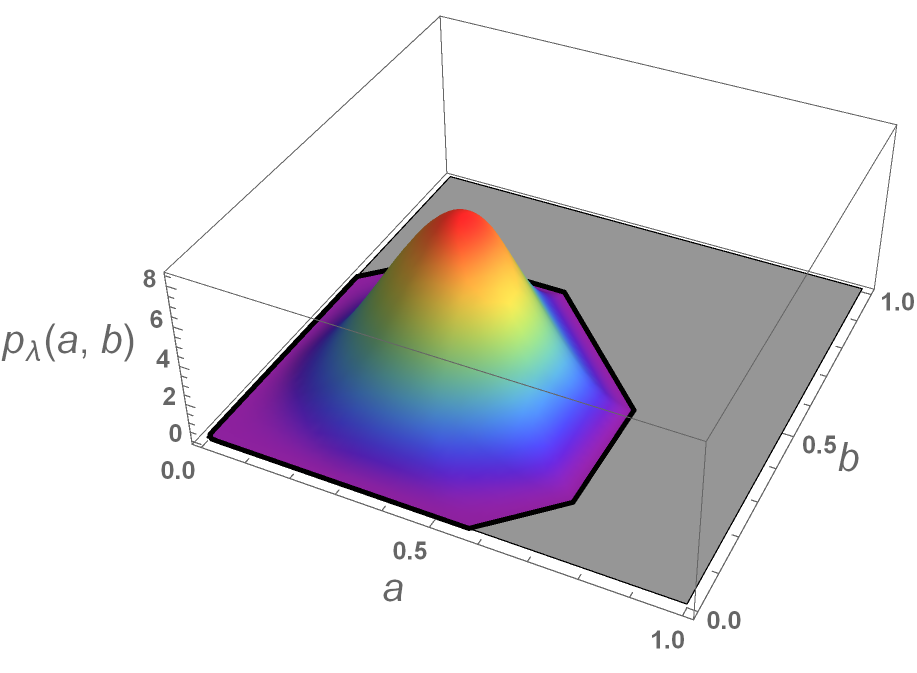}
\end{minipage}}
\caption{Joint two-qubit Bloch-radius density $p_{\bdlambda}(a,b)$
for $\bdlambda=(\frac47,\frac27,\frac17,0)$. Gray color is used
outside the support.}\label{fig:jointdensity}
\end{figure}

\begin{proof}
By Fourier inversion,
\begin{eqnarray*}
q_{\bdlambda}(\alpha,\beta) = \frac1{(2\pi)^2} \int_{\bbR^2}
\widehat q_{\bdlambda}(s,t)e^{-\mathrm{i}(s\alpha+t\beta)} \dif
s\dif t.
\end{eqnarray*}
Therefore,
\begin{eqnarray*}
\partial_a\partial_b q_{\bdlambda}(a,b)
= \cF^{-1} \Br{ (-\mathrm{i}s)(-\mathrm{i}t)\widehat
q_{\bdlambda}(s,t)}(a,b) = \cF^{-1} \Br{-st \widehat
q_{\bdlambda}(s,t)}(a,b).
\end{eqnarray*}
Substituting Proposition~\ref{prop:jointcharf} gives
\begin{eqnarray*}
-st \widehat q_{\bdlambda}(s,t) = \frac{3}{16V_4(\bdlambda)}
\sum_{\pi\in S_4} \sign(\pi)
\frac{e^{\mathrm{i}(su_\pi(\bdlambda)+tv_\pi(\bdlambda))}}{st(s+t)(s-t)}.
\end{eqnarray*}
By the translation property of the Fourier transform,
\begin{eqnarray*}
\partial_a\partial_b q_{\bdlambda}(a,b) = \frac{3}{16V_4(\bdlambda)}\sum_{\pi\in
S_4}\sign(\pi)G(a-u_\pi(\bdlambda),b-v_{\pi}(\bdlambda)).
\end{eqnarray*}
Independent local-unitary invariance implies that, conditional on
the radii $(a,b)$, the two directions are distributed according to
the unique $\SO(3)\times \SO(3)$-invariant probability measure on
$S^2\times S^2$, namely the product of the uniform spherical
measures. Hence
\begin{eqnarray*}
p_{\bdlambda}(a,b) = (4ab)\partial_a\partial_bq_{\bdlambda}(a,b),
\quad a,b>0.
\end{eqnarray*}
It follows that
\begin{eqnarray*}
p_{\bdlambda}(a,b) = \frac{3ab}{4V_4(\bdlambda)}\sum_{\pi\in
S_4}\sign(\pi)G(a-u_\pi(\bdlambda),b-v_{\pi}(\bdlambda))
\end{eqnarray*}
The complete alternating sum is nonnegative because it is obtained
from the push-forward of a probability measure. The density is
vanished outside $[0,1]^2$ is apparent because both $a$ and $b$ fall
in $[0,1]$.
\end{proof}

\begin{remark}
Although Eq.~\eqref{eq:jointdensityab} is an alternating sum, the
complete expression is non-negative. Non-negativity follows from its
construction as the push-forward density of a probability measure.
Individual spline summands need not be non-negative after
translation and alternation.
\end{remark}

\subsection{Support and the two-qubit quantum marginal polytope}

Let
\begin{eqnarray}\label{eq:muAB}
\mu_A=\frac{1-a}2,\quad \mu_B=\frac{1-b}2
\end{eqnarray}
be the smaller eigenvalues of $\rho_A$ and $\rho_B$. Bravyi's
compatibility conditions \cite{Bravyi2004} for a two-qubit state
with global spectrum
$\bdlambda=(\lambda_1,\lambda_2,\lambda_3,\lambda_4)$ are
\begin{eqnarray}\label{eq:Bravyi}
\begin{cases}
\min(\mu_A,\mu_B)\geqslant \lambda_3+\lambda_4,\\
\mu_A+\mu_B\geqslant \lambda_2+\lambda_3+2\lambda_4,\\
\abs{\mu_A-\mu_B}\leqslant
\min(\lambda_1-\lambda_3,\lambda_2-\lambda_4).
\end{cases}
\end{eqnarray}
In Bloch-radius coordinates, these equivalently become the
following. But now we need to derive it from Theorem~\ref{th:2qbit}.
To make the wall-scanning argument rigorous, we need the following
elementary algebraic fact in proving
Corollary~\ref{cor:supp2qubits}.
\begin{lem}[Alternating polynomial cancelation]\label{lem:vanish}
Let $\bdlambda\in C_4\cap\Delta_3$, and for each permutation $\pi\in
S_4$ define
\begin{eqnarray*}
\begin{cases}
u_\pi(\bdlambda)=\lambda_{\pi(1)}+\lambda_{\pi(2)}-\lambda_{\pi(3)}-\lambda_{\pi(4)},\\
v_\pi(\bdlambda)=\lambda_{\pi(1)}-\lambda_{\pi(2)}+\lambda_{\pi(3)}-\lambda_{\pi(4)}.
\end{cases}
\end{eqnarray*}
Then for any polynomial $P(u,v)$ of total degree $<6$,
\begin{eqnarray}\label{eq:vanish}
\sum_{\pi\in S_4}\sign(\pi)P(u_\pi(\bdlambda),v_\pi(\bdlambda))=0.
\end{eqnarray}
\end{lem}

\begin{proof}
The quantities $u_\pi(\bdlambda)$ and $v_\pi(\bdlambda)$ are affine
linear combinations of the four $\lambda_i$'s. Hence
$P(u_\pi(\bdlambda),v_\pi(\bdlambda))$, viewed as a function of
$\bdlambda=(\lambda_1,\ldots,\lambda_4)$, is a polynomial whose
total degree equals the total degree of $P$. Multiplying by
$\sign(\pi)$ and summing over $\pi$ makes the expression alternating
with respect to permutations of the $\lambda_i$'s, because replacing
$\bdlambda$ by
$\bdlambda_{\tau}=(\lambda_{\tau(1)},\ldots,\lambda_{\tau(4)})$
simply relabels the permutations. Every alternating polynomial is
divisible by the Vandermonde
$$
V_4(\bdlambda)=\prod_{1\leqslant i<j\leqslant
4}(\lambda_i-\lambda_j),
$$
which has degree $6$ (see, e.g., \cite[Proposition
10.21]{Hall2015}). Thus, if the total degree of $P$ is less than
$6$, the alternating sum must be the zero polynomial.
\end{proof}

\begin{cor}[Support of the joint Bloch radii density]\label{cor:supp2qubits}
The support of $p_{\bdlambda}(a,b)$ in Eq.~\eqref{eq:jointdensityab}
is the compact region
\begin{eqnarray}\label{eq:supportforqubits}
\op{supp}(p_{\bdlambda}) = \Set{(a,b)\in[0,1]^2:
\begin{cases}
\max(a,b)\leqslant 2(\lambda_1+\lambda_2)-1,\\
a+b\leqslant2(\lambda_1-\lambda_4),\\
\abs{a-b}\leqslant2\min(\lambda_1-\lambda_3,\lambda_2-\lambda_4)
\end{cases}}.
\end{eqnarray}
\end{cor}

\begin{proof}
Let $\Psi:\bsU(4)\to\bbR^2_{\geqslant0}$ be the continuous map
defined by
$$
\Psi(\bsU)=(a(\bsU),b(\bsU)),
$$
where $a(\bsU)$ and $b(\bsU)$ are the Bloch radii of the two
marginals of $\bsU\Lambda\bsU^\dagger$. Since
$\rho_{AB}=\bsU\Lambda\bsU^\dagger$ with $\bsU\sim \mu_{\haar}$ on
$\sfU(4)$, the push-forward measure is precisely
$\Psi_*\mu_{\haar}$. Because Haar measure has full support on the
compact connected Lie group $\sfU(4)$ \cite{Hall2015}, and $\Psi$ is
continuous, we have
$$
\op{supp}(\Psi_*\mu_{\haar})=\Psi(\sfU(4))(\subset[0,1]^2),
$$
which is precisely the image of the unitary orbit under the marginal
map in Bloch-radius coordinates. Indeed, for any open set
$O\subset\bbR^2$ with $O\cap\Psi(\sfU(4))\neq\emptyset$, there
exists $\bsU_0\in\sfU(4)$ with $\Psi(\bsU_0)\in O$. By continuity,
$\Psi^{-1}(O)$ is an open neighborhood of $\bsU_0$, hence has
positive Haar measure. Conversely, if
$O\cap\Psi(\sfU(4))=\emptyset$, then $\Psi^{-1}(O)=\emptyset$ and
the measure is zero on $O$.

Let $\Upsilon_{\bdlambda}(a,b):=\sum_{\pi\in
S_4}\sign(\pi)G(a-u_\pi(\bdlambda),b-v_\pi(\bdlambda))$. Thus
$\op{supp}(\Upsilon_{\bdlambda})=\op{supp}(p_{\bdlambda})$ since
$V_4(\bdlambda)>0$ for a strictly ordered spectrum, the closed
support of $p_{\bdlambda}$ is obtained from the closed support of
$\Upsilon_{\bdlambda}$ in the first quadrant, taking into account
that the factor $ab$ only makes the density vanish point-wise on the
axes and does not remove those axes from the closed support.

A direct check in the four regions defining $G$ gives
\begin{eqnarray}\label{eq:gxysupp}
G(x,y)=\frac{\chi_{\set{x\geqslant0}}}4\Br{(x+y)^2_+ - 2y^2_+ +
(y-x)^2_+},
\end{eqnarray}
whose support is the cone $\op{supp}(G) =
\Set{(x,y)\in\bbR^2:x\geqslant0,x+y\geqslant0}$. Define
\begin{eqnarray*}
\begin{cases}
\tilde\alpha&:=2(\lambda_1+\lambda_2)-1,\\
\tilde\beta&:=2(\lambda_1+\lambda_3)-1,\\
\tilde\gamma&:=2(\lambda_1+\lambda_4)-1.
\end{cases}
\end{eqnarray*}
It is easily seen that
$$
1>\tilde\alpha>\tilde\beta>\abs{\tilde\gamma}\geqslant0.
$$
The essential point is that we need not only the $24$ points but
also their signs. Let $\cK$ denote the set of $24$ signed knots, and
let $\varepsilon(u,v)\in\set{\pm1}$ be the displayed sign, see
Appendix~\ref{app:lsit24pts}. Using Eq.~\eqref{eq:gxysupp}, the
alternating sum becomes the completely explicit expression
\begin{eqnarray}\label{eq:Flambdaab}
4\Upsilon_{\bdlambda}(a,b)=\sum_{(u,v)\in
\cK}\varepsilon(u,v)\chi_{\set{a\geqslant u}}\Br{(a+b-u-v)^2_+ -
2(b-v)^2_+ + (b-a+u-v)^2_+}.
\end{eqnarray}
Everything about the support can now be read from the cancelations
among these quadratic hinge functions. Indeed, on each chamber cut
out by the lines
\begin{eqnarray}\label{eq:walls}
a=u,\quad b=v,\quad a+b=u+v,\quad a-b=u-v.
\end{eqnarray}
the argument $(a-u,b-v)$ lies in a fixed region of the piecewise
definition of the bivariate box-spline $G$, so that $G(a-u,b-v)$
reduces to a quadratic polynomial in the variables $u$ and $v$ (with
coefficients depending on $a,b$). Consequently, for any chamber that
avoids the walls, the alternating sum satisfies the hypothesis of
Lemma~\ref{lem:vanish} and hence vanishes identically. Therefore the
support of $\Upsilon_{\bdlambda}$---and thus of
$p_{\bdlambda}$---can only be non-zero after crossing one of those
walls.

The following wall-scanning procedure systematically identifies the
first walls where the signed cancelation fails; these walls are
precisely the boundaries of the Bravyi-Klyachko polytope. Starting
from the exterior region---where the density vanishes identically
because the alternating sum of translated splines cancels completely
(Lemma~\ref{lem:vanish})---one moves inward across the candidate
walls Eq.~\eqref{eq:walls}. The support boundary is reached at the
first wall (in the inward direction) where the cancellation is no
longer complete; that is, where at least one translated spline term
changes its polynomial branch, thereby breaking the alternating
cancelation. Scanning all possible directions in this manner yields
exactly the inequalities Eq.~\eqref{eq:supportforqubits}. The
computing details can be described as follows.
\begin{itemize}
\item \textbf{The upper bound $a\leqslant\tilde\alpha=2(\lambda_1+\lambda_2)-1$.} The largest first coordinate among all knots is $\tilde\alpha$.
The knots with $u=\tilde\alpha$, including their signs, are
\begin{eqnarray}\label{eq:foursignedknots}
(\tilde\alpha,\tilde\beta)^+,\quad
(\tilde\alpha,-\tilde\beta)^+,\quad
(\tilde\alpha,\tilde\gamma)^-,\quad (\tilde\alpha,-\tilde\gamma)^-.
\end{eqnarray}
For $a>\tilde\alpha$, all factors $\chi_{\set{a\geqslant u}}$ in
Eq.~\eqref{eq:Flambdaab} are active. On every chamber determined by
the remaining hinge walls, the expression is then an alternating sum
of quadratic polynomials. By Eq.~\eqref{eq:vanish}, the polynomial
coefficients cancel.

Scanning the $b$-walls in decreasing order gives zero in every
chamber with $a>\tilde\alpha$. Thus
$$
\Upsilon_{\bdlambda}(a,b)=0\quad\text{for
}a>\tilde\alpha,b\geqslant0.
$$
At $a=\tilde\alpha$, the four knots in
Eq.~\eqref{eq:foursignedknots} are precisely the final knots whose
indicator functions change. On the side $a<\tilde\alpha$, the
cancelation is no longer complete. Therefore the vertical support
boundary is $a=\tilde\alpha$. The analogous wall scan in the second
coordinate gives
$$
\Upsilon_{\bdlambda}(a,b)=0\quad\text{for }a\geqslant0,
b>\tilde\alpha.
$$
and the horizontal support boundary is $b=\tilde\alpha$. Hence, in
the first quadrant,
$$
a\leqslant\tilde\alpha,\quad
b\leqslant\tilde\alpha\Longleftrightarrow \max(a,b)\leqslant
\tilde\alpha=2(\lambda_1+\lambda_2)-1.
$$
In other words, $\Upsilon_{\bdlambda}(a,b)=0$ outside the square
$[0,\tilde\alpha]^2(\subset[0,1]^2)$ in the first quadrant.
\item \textbf{The upper bound $a+b\leqslant\tilde\alpha+\tilde\beta=2(\lambda_1-\lambda_4)$.} The first hinge in Eq.~\eqref{eq:Flambdaab},
$(a+b-u-v)^2_+$, changes branch on the lines $a+b=u+v$. The largest
possible value of $u+v$ is $\tilde\alpha+\tilde\beta$. The two knots
at this level are $(\tilde\alpha,\tilde\beta)^+$ and
$(\tilde\beta,\tilde\alpha)^-$. Their contribution to the first
hinge is $\Pa{\chi_{\set{a\geqslant\tilde\alpha}} -
\chi_{\set{a\geqslant\tilde\beta}}}(a+b-\tilde\alpha-\tilde\beta)^2_+$.
In the relevant strip $\tilde\beta<a<\tilde\alpha$, the first
indicator is zero and the second is one. Thus the cancelation
changes exactly across $a+b=\tilde\alpha+\tilde\beta$. When
$a+b>\tilde\alpha+\tilde\beta$, the contributions from all three
hinge families in Eq.~\eqref{eq:Flambdaab} cancel after the signed
knots are collected; when one crosses into
$a+b<\tilde\alpha+\tilde\beta$ with $\tilde\beta<a<\tilde\alpha$,
the cancelation fails. Therefore
\begin{eqnarray*}
\Upsilon_{\bdlambda}(a,b)=0 \quad\text{if
}a+b>\tilde\alpha+\tilde\beta,
\end{eqnarray*}
and the oblique upper support boundary is
$a+b=\tilde\alpha+\tilde\beta= 2(\lambda_1-\lambda_4)$. Thus this
inequality will later become
\begin{eqnarray*}
a+b\leqslant 2(\lambda_1-\lambda_4).
\end{eqnarray*}
\item \textbf{The upper bound $\abs{a-b}\leqslant
\tilde\alpha-\abs{\tilde\gamma}=2\min(\lambda_1-\lambda_3,\lambda_2-\lambda_4)$.}
The last hinge in Eq.~\eqref{eq:Flambdaab}, $(b-a+u-v)^2_+$, changes
branch on $a-b=u-v$. At first sight, the largest value of $u-v$ is
$\tilde\alpha+\tilde\beta$, but the corresponding walls cancel. This
is the important point that the convex hull of the $24$ knots alone
does not detect.

We scan the positive $a-b$ direction.
\begin{enumerate}[(i)]
\item \textbf{The apparent outer wall $a-b=\tilde\alpha+\tilde\beta$.} The knots
with $u-v=\tilde\alpha+\tilde\beta$ are
$(\tilde\alpha,-\tilde\beta)^+$ and $(\tilde\beta,-\tilde\alpha)^-$.
Their contribution to the last hinge is
$\Pa{\chi_{\set{a\geqslant\tilde\alpha}} -
\chi_{\set{a\geqslant\tilde\beta}}}(b-a+\tilde\alpha+\tilde\beta)^2_+$.
If $a\geqslant\tilde\alpha$, the two terms cancel. If
$\tilde\beta\leqslant a<\tilde\alpha$, then near the line
$a-b=\tilde\alpha+\tilde\beta$,
$b=a-(\tilde\alpha+\tilde\beta)=(a-\tilde\alpha)-\tilde\beta<-\tilde\beta<0$.
Thus the uncanceled part of this wall lies entirely outside the
first quadrant. Hence $a-b=\tilde\alpha+\tilde\beta$ is not a
support boundary in $\bbR^2_{\geqslant0}$.
\item \textbf{The next apparent wall.} (1) First suppose that
$\tilde\gamma\geqslant0$. The next positive level is
$u-v=\tilde\alpha+\tilde\gamma$. It comes from
$(\tilde\alpha,-\tilde\gamma)^-$ and
$(\tilde\gamma,-\tilde\alpha)^+$. Their last-hinge contribution is
$\Pa{-\chi_{\set{a\geqslant\tilde\alpha}} +
\chi_{\set{a\geqslant\tilde\gamma}}}(b-a+\tilde\alpha+\tilde\gamma)^2_+$.
For $\tilde\gamma\leqslant a<\tilde\alpha$, the second term remains.
But on the corresponding wall,
$b=a-(\tilde\alpha+\tilde\gamma)=(a-\tilde\alpha)-\tilde\gamma<-\tilde\gamma\leqslant0$.
Except possibly at an endpoint, this wall is again outside the first
quadrant. Therefore it does not bound the support in the first
quadrant. (2) If $\tilde\gamma<0$, the analogous canceled outer
level is
$\tilde\alpha-\tilde\gamma=\tilde\alpha+\abs{\tilde\gamma}$. It
comes from $(\tilde\alpha,\tilde\gamma)^-$ and
$(-\tilde\gamma,-\tilde\alpha)^+$, and the same argument shows that
its uncanceled portion lies outside $b\geqslant0$. Thus, in either
case, the wall $a-b=\tilde\alpha+\abs{\tilde\gamma}$ does not
contribute to the first-quadrant support.
\item \textbf{The first wall that survives in the first quadrant.}
(1) Again suppose first that $\tilde\gamma\geqslant0$. At the next
level, $u-v=\tilde\alpha-\tilde\gamma$, we have the pair
$(\tilde\alpha,\tilde\gamma)^-$ and
$(-\tilde\gamma,-\tilde\alpha)^+$. Their contribution is
$\Pa{-\chi_{\set{a\geqslant\tilde\alpha}} +
\chi_{\set{a\geqslant-\tilde\gamma}}}
(b-a+\tilde\alpha-\tilde\gamma)^2_+$. In the first quadrant,
$a\geqslant0\geqslant-\tilde\gamma$, so the second indicator is
already active. For $a<\tilde\alpha$, the first indicator is
inactive. Therefore this hinge genuinely survives. The wall is
$a-b=\tilde\alpha-\tilde\gamma$. Unlike the previous walls, it meets
the first quadrant in the segment
$$
\tilde\alpha-\tilde\gamma\leqslant a\leqslant\tilde\alpha, \quad
0\leqslant b\leqslant\tilde\gamma.
$$
Thus this is a genuine support boundary. (2) If $\tilde\gamma<0$,
the same analysis uses the pair $(\tilde\alpha,-\tilde\gamma)^-$ and
$(\tilde\gamma,-\tilde\alpha)^+$, and the surviving wall is
$a-b=\tilde\alpha+\tilde\gamma=\tilde\alpha-\abs{\tilde\gamma}$.
Both cases can therefore be written as
$a-b\leqslant\tilde\alpha-\abs{\tilde\gamma}$. Performing the
corresponding wall scan in the opposite diagonal direction yields
$b-a\leqslant\tilde\alpha-\abs{\tilde\gamma}$. Together,
\begin{eqnarray*}
\abs{a-b}\leqslant\tilde\alpha-\abs{\tilde\gamma}=2\min(\lambda_1-\lambda_3,\lambda_2-\lambda_4).
\end{eqnarray*}
This is precisely the extra cancelation condition that does not
follow from the convex hull of the $24$ knots.
\end{enumerate}
\end{itemize}
Note that Eq.~\eqref{eq:Flambdaab} divides the first quadrant into
finitely many chambers bounded by lines of the four types
$$
a=u,\quad b=v,\quad a+b=u+v,\quad a-b=u-v,
$$
where $(u,v)$ runs over the signed knot list, see Table~\ref{tab:2}
in Appendix~\ref{app:lsit24pts}. On each chamber, every indicator
and positive-part function in Eq.~\eqref{eq:Flambdaab} has a fixed
branch. Hence $\Upsilon_{\bdlambda}$ is a quadratic polynomial
there.

Starting from the exterior chambers, where the
alternating-polynomial cancelation Eq.~\eqref{eq:vanish} gives zero,
the preceding wall scan gives the first uncanceled walls:
\begin{eqnarray*}
a=\tilde\alpha,\quad b=\tilde\alpha,\quad
a+b=\tilde\alpha+\tilde\beta,\quad
a-b=\tilde\alpha-\abs{\tilde\gamma}, \quad
b-a=\tilde\alpha-\abs{\tilde\gamma}.
\end{eqnarray*}
Continuing the same finite collection through the internal walls
changes the quadratic formula for $\Upsilon_{\bdlambda}$, but it
does not make that polynomial identically zero on any chamber
satisfying all five strict inequalities. Thus $\Upsilon_{\bdlambda}
\not\equiv0$ on every open chamber contained in $0<a<\tilde\alpha$
and $0<b<\tilde\alpha$,
$$
a+b<\tilde\alpha+\tilde\beta, \quad
\abs{a-b}<\tilde\alpha-\abs{\tilde\gamma}.
$$
Since $p_{\bdlambda}$ is a nonnegative density and
$$
p_{\bdlambda}(a,b) = \frac{3ab}{4V_4(\bdlambda)}
\Upsilon_{\bdlambda}(a,b),
$$
it follows that every such chamber contains points of positive
density, which accounts for the absence of further open zero
regions. Taking closures yields the support, and the direct
cancelation calculation consequently establishes the desired result.
\end{proof}
Thus the support of the probabilistic distribution recovers the
deterministic two-qubit marginal polytope, while the density
Eq.~\eqref{eq:jointdensityab} gives the probability weight inside
that polytope. This conclusion has been obtained from the translated
truncated-power formula for $G$ and the signed $S_4$-sum, without
assuming marginal compatibility inequalities.
\begin{remark}
The support of $p_{\bdlambda}$ is the convex polygon with vertices,
in counterclockwise order,
$$
(0,0),(\tilde\alpha-\abs{\tilde\gamma},0),(\tilde\alpha,\abs{\tilde\gamma}),(\tilde\alpha,\tilde\beta),(\tilde\beta,\tilde\alpha),
(\abs{\tilde\gamma},\tilde\alpha),(0,\tilde\alpha-\abs{\tilde\gamma}).
$$
\begin{itemize}
\item For $\tilde\gamma\neq0$, these are seven distinct vertices, so the
support is a heptagon:
$$
\op{supp}(p_{\bdlambda})=\op{Conv}\Set{(0,0),(\tilde\alpha-\abs{\tilde\gamma},0),(\tilde\alpha,\abs{\tilde\gamma}),(\tilde\alpha,\tilde\beta),(\tilde\beta,\tilde\alpha),
(\abs{\tilde\gamma},\tilde\alpha),(0,\tilde\alpha-\abs{\tilde\gamma})}.
$$
\item For $\tilde\gamma=0$, i.e.,
$\lambda_1+\lambda_4=\lambda_2+\lambda_3=\frac12$, the support
reduces to the pentagon
$$
\op{supp}(p_{\bdlambda})=\Set{(0,0),
(\tilde\alpha,0),(\tilde\alpha,\tilde\beta),(\tilde\beta,\tilde\alpha),(0,\tilde\alpha)}.
$$
\end{itemize}
\end{remark}

\subsection{Distribution of one two-qubit marginal}

A compact truncated-power formula can also be obtained for one
marginal. For a two-element subset
$J=\set{(i,j):i<j}\subset\set{1,2,3,4}$, define
\begin{eqnarray}\label{eq:VJ}
V_J(\bdlambda)=\lambda_i-\lambda_j,
\end{eqnarray}
and let $J^c$ be its complement. Set
\begin{eqnarray}\label{eq:uJ}
w_J(\bdlambda)&=&2\sum_{j\in J}\lambda_j-1,\\
\varepsilon_J&=&(-1)^{i+j+1}.
\end{eqnarray}
Let $q^A_{\bdlambda}(\alpha)$ denote the density of the fixed Bloch
component $\alpha=\Tr{\rho_A\sigma_3}$.
\begin{prop}
The characteristic function and density of $\alpha$ are
\begin{eqnarray}
\widehat q^A_{\bdlambda}(s) &=&
\frac3{4V_4(\bdlambda)}\sum_{\abs{J}=2}\varepsilon_J
V_J(\bdlambda)V_{J^c}(\bdlambda)\frac{e^{\mathrm{i}sw_J(\bdlambda)}}{(\mathrm{i}s)^4},\label{eq:widehatq}\\
q^A_{\bdlambda}(\alpha) &=&
\frac1{8V_4(\bdlambda)}\sum_{\abs{J}=2}\varepsilon_J
V_J(\bdlambda)V_{J^c}(\bdlambda)(w_J(\bdlambda)-\alpha)^3_+.\label{eq:qAalpha}
\end{eqnarray}
\end{prop}

\begin{proof}
For $\bsH_s=s\sigma_3\ot\I_2$, its eigenvalues are $s,s,-s,-s$.
Applying the double-confluent limit to the $\sfU(4)$-HCIZ integral
gives
$$
\widehat
q^A_{\bdlambda}(s)=\frac3{4V_4(\bdlambda)}\frac1{(\mathrm{i}s)^4}\sum_{\abs{J}=2}\varepsilon_JV_J(\bdlambda)V_{J^c}(\bdlambda)e^{\mathrm{i}sw_J(\bdlambda)}.
$$
The inverse-transform identity
$$
\cF^{-1}\Pa{\frac{e^{\mathrm{i}sw}}{(\mathrm{i}s)^k}}(x)=\frac{(w-x)^{k-1}_+}{(k-1)!}
$$
with $k=4$ gives Eq.~\eqref{eq:qAalpha}.
\end{proof}

\begin{thrm}[One-marginal Bloch-radius density]\label{th:1marginalBlochr}
The density of the marginal Bloch radius $a$ is
\begin{eqnarray}\label{eq:pa}
p^A_{\bdlambda}(a)=\frac{3a}{4V_4(\bdlambda)}\sum_{\abs{J}=2}\varepsilon_J
V_J(\bdlambda)V_{J^c}(\bdlambda)(w_J(\bdlambda)-a)^2_+
\end{eqnarray}
for $a\in\bbR_{\geqslant0}$. Its support is the closed interval
$[0,2(\lambda_1+\lambda_2)-1]$.
\end{thrm}

\begin{proof}
Using the $\SU(2)$-derivative principle,
$p^A_{\bdlambda}(a)=-2a\frac{\dif q^A_{\bdlambda}(a)}{\dif a}$.
Differentiating Eq.~\eqref{eq:qAalpha} gives Eq.~\eqref{eq:pa}.
\end{proof}

Let $a$ be the Bloch radius of marginal state of $A$ subsystem and
let $p^A_{\bdlambda}$ be its density from
Theorem~\ref{th:1marginalBlochr}. The ordered marginal eigenvalues
are $z_\pm=\frac{1\pm a}2$. Their probability densities are
$$
\begin{cases}
f(z_+)=2p^A_{\bdlambda}(2z_+-1),&\text{if } \frac12\leqslant
z_+\leqslant\lambda_1+\lambda_2,\\
f(z_-)=2p^A_{\bdlambda}(1-2z_-),&\text{if }
\lambda_3+\lambda_4\leqslant z_-\leqslant\frac12.
\end{cases}
$$
If $z$ is obtained by choosing one of the two marginal eigenvalues
uniformly at random, then
\begin{eqnarray*}
P^A_{\bdlambda}(z)=p^A_{\bdlambda}(\abs{2z-1}),\quad
z\in[\lambda_3+\lambda_4,\lambda_1+\lambda_2].
\end{eqnarray*}
In particular,
\begin{eqnarray*}
P^A_{\bdlambda}(z)=\frac{3\abs{2z-1}}{4V_4(\bdlambda)}\sum_{\abs{J}=2}\varepsilon_JV_{J^c}(\bdlambda)V_J(\bdlambda)(w_J(\bdlambda)-\abs{2z-1})^2_+.
\end{eqnarray*}
We can also directly derive another form equivalent to that given
above.

\begin{cor}[Marginal eigenvalue distributions]\label{th:1sided}
For a random two-qubit state $\rho_{AB}\in \cU_\Lambda$, where
$\Lambda$ corresponds to $\bdlambda\in C_4\cap\Delta_3$, denote
\begin{eqnarray}
c_1=\lambda_1+\lambda_2,c_2=\lambda_1+\lambda_3,c_3=\max(\lambda_1+\lambda_4,\lambda_2+\lambda_3),\\
c_4=\min(\lambda_1+\lambda_4,\lambda_2+\lambda_3),c_5=\lambda_2+\lambda_4,c_6=\lambda_3+\lambda_4.
\end{eqnarray}
The distribution density of a generic eigenvalue $z$ of $\rho_A$ is
piecewise polynomially, given by
\begin{eqnarray}
P^A_{\bdlambda}(z)
=2\sum^5_{k=1}f^{(k)}_{\bdlambda}(z)\chi_{[c_{k+1},c_k]}(z),
\end{eqnarray}
which is shown in Figure~\ref{fig:genericpdf}, where
\begin{eqnarray}
f^{(1)}_{\bdlambda}(z) &=&
\frac{(\lambda_1+\lambda_2-z)^3}{\prod^2_{i=1}\prod^4_{j=3}(\lambda_i-\lambda_j)},\\
f^{(2)}_{\bdlambda}(z) &=& \frac{F_3(\bdlambda)z^3+F_2(\bdlambda)z^2+F_1(\bdlambda)z+F_0(\bdlambda)}{\prod^3_{i=2}(\lambda_i-\lambda_4)\prod^4_{j=2}(\lambda_1-\lambda_j)}, \\
f^{(3)}_{\bdlambda}(z) &=& \frac{-3z^2+3z+F(\bdlambda)}{\prod^4_{j=2}(\lambda_1-\lambda_j)}\text{ or }\frac{-3z^2+3z+F(\bdlambda^{(14)})}{\prod^3_{i=1}(\lambda_i-\lambda_4)}, \\
f^{(4)}_{\bdlambda}(z) &=& \frac{F_3(\bdlambda^{(14)})z^3+F_2(\bdlambda^{(14)})z^2+F_1(\bdlambda^{(14)})z+F_0(\bdlambda^{(14)})}{\prod^3_{i=2}(\lambda_i-\lambda_4)\prod^4_{j=2}(\lambda_1-\lambda_j)}, \\
f^{(5)}_{\bdlambda}(z) &=&
\frac{(z-(\lambda_3+\lambda_4))^3}{\prod^2_{i=1}\prod^4_{j=3}(\lambda_i-\lambda_j)}.
\end{eqnarray}
In the above, $\bdlambda=(\lambda_1,\lambda_2,\lambda_3,\lambda_4)$
and $\bdlambda^{(14)}=(\lambda_4,\lambda_2,\lambda_3,\lambda_1)$,
and
\begin{eqnarray*}
F_3(\bdlambda)&=& \lambda_1-\lambda_4,\\
F_2(\bdlambda)&=& -3\Br{\lambda_1^2-(\lambda_2+\lambda_3)\lambda_4+\lambda_2\lambda_3},\\
F_1(\bdlambda)&=& 3\Br{\lambda^3_1+(\lambda_1\lambda_4+\lambda_2\lambda_3)\lambda_1-(\lambda^2_2+\lambda^2_3+\lambda_2\lambda_3)\lambda_4-(\lambda_2+\lambda_3)\lambda_1\lambda_4+\lambda_2\lambda_3(\lambda_2+\lambda_3)},\\
F_0(\bdlambda) &=& -\lambda_1^4-2\lambda_4 \lambda_1^3-2
(\lambda_2+\lambda_3)\lambda_2\lambda_3\lambda_1+(\lambda_2+\lambda_3)(\lambda^2_2+\lambda^2_3)\lambda_4 \\
&&+2(
\lambda^2_2+\lambda^2_3+\lambda_2\lambda_3)\lambda_1\lambda_4-\lambda_2\lambda_3(\lambda^2_2+\lambda_2\lambda_3+\lambda^2_3),\\
F(\bdlambda)&=&\lambda^2_1+\lambda_2\lambda_3+\lambda_2\lambda_4+\lambda_3\lambda_4-1.
\end{eqnarray*}
The expression of $f^{(3)}_{\bdlambda}(z)$ should be
$\frac{-3z^2+3z+F(\bdlambda)}{\prod^4_{j=2}(\lambda_1-\lambda_j)}$
if $\lambda_1+\lambda_4>\lambda_2+\lambda_3$; and
$\frac{-3z^2+3z+F(\bdlambda^{(14)})}{\prod^3_{i=1}(\lambda_i-\lambda_4)}$
if $\lambda_1+\lambda_4<\lambda_2+\lambda_3$. The probability
density curve of a generic eigenvalue of $\rho_B$ is the same as
that of $\rho_A$.
\end{cor}

\begin{figure}[h!]\centering
{\begin{minipage}[b]{0.5\linewidth}
\includegraphics[width=1\textwidth]{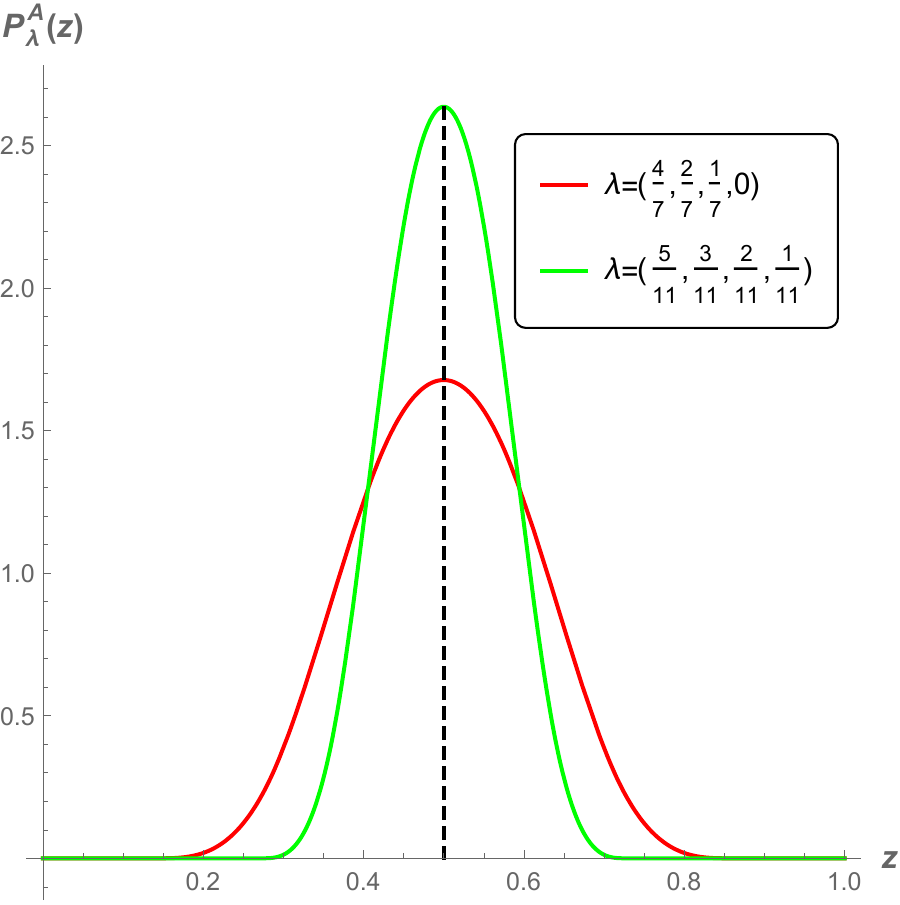}
\end{minipage}}
\caption{The probability density curve of a generic eigenvalue of a
marginal state of a two-qubit unitary orbit.}\label{fig:genericpdf}
\end{figure}

\begin{proof}
For $N=4=mn$, where $m=n=2$, let
$\bdlambda=(\lambda_1,\ldots,\lambda_4)\in C_4\cap\Delta_3$. Then we
have
\begin{eqnarray}
c_1> c_2> c_3\geqslant\frac12\geqslant c_4> c_5> c_6.
\end{eqnarray}
Moreover, $c_1+c_6=c_2+c_5=c_3+c_4=1$. Let
$\bss=\diag(s_1,s_2,s_3,s_4)\to\diag(x_1,x_2,x_1,x_2)$, it holds
that
\begin{eqnarray*}
&&\varphi_{\bdlambda}(\bsx) =
\mathrm{i}^6\frac{\prod^4_{k=1}\Gamma(k)}{V_4(\bdlambda)}\lim_{(s_3,s_4)\to
(x_1,x_2)}\lim_{(s_1,s_2)\to
(x_1,x_2)}\frac{\det_4(e^{\mathrm{i}s_i\lambda_j})}{V_4(\bss)}\\
&&=-\frac{12}{V_4(\bdlambda)}\lim_{(s_3,s_4)\to
(x_1,x_2)}\lim_{(s_1,s_2)\to
(x_1,x_2)}\frac1{(s_1-s_2)(s_1-s_4)(s_2-s_3)(s_3-s_4)}\\
&&~~~\Abs{\begin{array}{cccc}
         e^{\mathrm{i}s_1\lambda_1} & e^{\mathrm{i}s_1\lambda_2} & e^{\mathrm{i}s_1\lambda_3} & e^{\mathrm{i}s_1\lambda_4} \\
         e^{\mathrm{i}s_2\lambda_1} & e^{\mathrm{i}s_2\lambda_2} & e^{\mathrm{i}s_2\lambda_3} & e^{\mathrm{i}s_2\lambda_4} \\
        \frac{e^{\mathrm{i}s_3\lambda_1}-e^{\mathrm{i}s_1\lambda_1}}{s_3-s_1} & \frac{e^{\mathrm{i}s_3\lambda_2}-e^{\mathrm{i}s_1\lambda_2}}{s_3-s_1} & \frac{e^{\mathrm{i}a_3\lambda_3}-e^{\mathrm{i}a_1\lambda_3}}{a_3-a_1} & \frac{e^{\mathrm{i}a_3\lambda_4}-e^{\mathrm{i}a_1\lambda_4}}{a_3-a_1} \\
         \frac{e^{\mathrm{i}s_4\lambda_1}-e^{\mathrm{i}s_2\lambda_1}}{s_4-s_2} & \frac{e^{\mathrm{i}s_4\lambda_2}-e^{\mathrm{i}s_2\lambda_2}}{s_4-s_2} & \frac{e^{\mathrm{i}a_4\lambda_3}-e^{\mathrm{i}a_2\lambda_3}}{a_4-a_2} &
         \frac{e^{\mathrm{i}s_4\lambda_4}-e^{\mathrm{i}s_2\lambda_4}}{s_4-s_2}
       \end{array}
}\\
&&=-\frac{12}{V_4(\bdlambda)}\frac1{(x_1-x_2)^4}\Abs{\begin{array}{cccc}
         e^{\mathrm{i}x_1\lambda_1} & e^{\mathrm{i}x_1\lambda_2} & e^{\mathrm{i}x_1\lambda_3} & e^{\mathrm{i}x_1\lambda_4} \\
         e^{\mathrm{i}x_2\lambda_1} & e^{\mathrm{i}x_2\lambda_2} & e^{\mathrm{i}x_2\lambda_3} & e^{\mathrm{i}x_2\lambda_4} \\
         \lambda_1e^{\mathrm{i}x_1\lambda_1} & \lambda_2e^{\mathrm{i}x_1\lambda_2} & \lambda_3e^{\mathrm{i}x_1\lambda_3} & \lambda_4e^{\mathrm{i}x_1\lambda_4} \\
         \lambda_1e^{\mathrm{i}x_2\lambda_1} & \lambda_2e^{\mathrm{i}x_2\lambda_2} & \lambda_3e^{\mathrm{i}x_2\lambda_3} & \lambda_4e^{\mathrm{i}x_2\lambda_4}
       \end{array}
},
\end{eqnarray*}
where $\bsx=(x_1,x_2)$. Then
\begin{eqnarray*}
&&\Abs{\begin{array}{cccc}
         e^{\mathrm{i}x_1\lambda_1} & e^{\mathrm{i}x_1\lambda_2} & e^{\mathrm{i}x_1\lambda_3} & e^{\mathrm{i}x_1\lambda_4} \\
         e^{\mathrm{i}x_2\lambda_1} & e^{\mathrm{i}x_2\lambda_2} & e^{\mathrm{i}x_2\lambda_3} & e^{\mathrm{i}x_2\lambda_4} \\
         \lambda_1e^{\mathrm{i}x_1\lambda_1} & \lambda_2e^{\mathrm{i}x_1\lambda_2} & \lambda_3e^{\mathrm{i}x_1\lambda_3} & \lambda_4e^{\mathrm{i}x_1\lambda_4} \\
         \lambda_1e^{\mathrm{i}x_2\lambda_1} & \lambda_2e^{\mathrm{i}x_2\lambda_2} & \lambda_3e^{\mathrm{i}x_2\lambda_3} & \lambda_4e^{\mathrm{i}x_2\lambda_4}
       \end{array}
} \\
&&=
\lambda_1\lambda_2\Br{e^{\mathrm{i}(\lambda_1x_2+\lambda_2x_1)}-e^{\mathrm{i}(\lambda_1
x_1+\lambda_2 x_2)}}\Br{e^{\mathrm{i}(\lambda_3 x_2+\lambda_4
x_1)}-e^{\mathrm{i}(\lambda_3x_1+\lambda_4 x_2)}}\\
&&~~~-\lambda_1\lambda_3\Br{e^{\mathrm{i}(\lambda_1x_2+\lambda_3x_1)}-e^{\mathrm{i}
(\lambda_1x_1+\lambda_3 x_2)}}\Br{e^{\mathrm{i}(\lambda_2x_2+
\lambda_4x_1)}-e^{\mathrm{i}(\lambda_2x_1+\lambda_4x_2)}}\\
&&~~~+\lambda_1\lambda_4\Br{e^{\mathrm{i}(\lambda_2 x_2+\lambda_3
x_1)}-e^{\mathrm{i}(\lambda_2x_1+\lambda_3x_2)}}\Br{e^{\mathrm{i}(\lambda_1
x_2+\lambda_4 x_1)}-e^{\mathrm{i}(\lambda_1 x_1+\lambda_4x_2)}}\\
&&~~~+\lambda_2\lambda_3\Br{e^{\mathrm{i}(\lambda_2 x_2+\lambda_3
x_1)}-e^{\mathrm{i}(\lambda_2 x_1+\lambda_3
x_2)}}\Br{e^{\mathrm{i}(\lambda_1 x_2+\lambda_4
x_1)}-e^{\mathrm{i}(\lambda_1 x_1+\lambda_4 x_2)}}\\
&&~~~-\lambda_2\lambda_4\Br{e^{\mathrm{i}(\lambda_1 x_2+\lambda_3
x_1)}-e^{\mathrm{i}(\lambda_1x_1+\lambda_3x_2)}}\Br{e^{\mathrm{i}(\lambda_2
x_2+\lambda_4 x_1)}-e^{\mathrm{i}(\lambda_2 x_1+\lambda_4 x_2)}}\\
&&~~~+\lambda_3\lambda_4\Br{e^{\mathrm{i}(\lambda_1 x_2+\lambda_2
x_1)}-e^{\mathrm{i}(\lambda_1 x_1+\lambda_2
x_2)}}\Br{e^{\mathrm{i}(\lambda_3 x_2+\lambda_4
x_1)}-e^{\mathrm{i}(\lambda_3 x_1+\lambda_4 x_2)}}.
\end{eqnarray*}
This implies that
\begin{eqnarray*}
-\frac1{12}\varphi_{\bdlambda}(\bsx) &=&
\frac{\lambda_1\lambda_2}{V_4(\bdlambda)}\frac{\Br{e^{\mathrm{i}(\lambda_1x_2+\lambda_2x_1)}-e^{\mathrm{i}(\lambda_1
x_1+\lambda_2 x_2)}}\Br{e^{\mathrm{i}(\lambda_3 x_2+\lambda_4
x_1)}-e^{\mathrm{i}(\lambda_3x_1+\lambda_4 x_2)}}}{(x_1-x_2)^4}\\
&&-\frac{\lambda_1\lambda_3}{V_4(\bdlambda)}\frac{\Br{e^{\mathrm{i}(\lambda_1x_2+\lambda_3x_1)}-e^{\mathrm{i}
(\lambda_1x_1+\lambda_3 x_2)}}\Br{e^{\mathrm{i}(\lambda_2x_2+
\lambda_4x_1)}-e^{\mathrm{i}(\lambda_2x_1+\lambda_4x_2)}}}{(x_1-x_2)^4}\\
&&+\frac{\lambda_1\lambda_4}{V_4(\bdlambda)}\frac{\Br{e^{\mathrm{i}(\lambda_2
x_2+\lambda_3
x_1)}-e^{\mathrm{i}(\lambda_2x_1+\lambda_3x_2)}}\Br{e^{\mathrm{i}(\lambda_1
x_2+\lambda_4 x_1)}-e^{\mathrm{i}(\lambda_1 x_1+\lambda_4x_2)}}}{(x_1-x_2)^4}\\
&&+\frac{\lambda_2\lambda_3}{V_4(\bdlambda)}\frac{\Br{e^{\mathrm{i}(\lambda_2
x_2+\lambda_3 x_1)}-e^{\mathrm{i}(\lambda_2 x_1+\lambda_3
x_2)}}\Br{e^{\mathrm{i}(\lambda_1 x_2+\lambda_4
x_1)}-e^{\mathrm{i}(\lambda_1 x_1+\lambda_4 x_2)}}}{(x_1-x_2)^4}\\
&&-\frac{\lambda_2\lambda_4}{V_4(\bdlambda)}\frac{\Br{e^{\mathrm{i}(\lambda_1
x_2+\lambda_3
x_1)}-e^{\mathrm{i}(\lambda_1x_1+\lambda_3x_2)}}\Br{e^{\mathrm{i}(\lambda_2
x_2+\lambda_4 x_1)}-e^{\mathrm{i}(\lambda_2 x_1+\lambda_4 x_2)}}}{(x_1-x_2)^4}\\
&&+\frac{\lambda_3\lambda_4}{V_4(\bdlambda)}\frac{\Br{e^{\mathrm{i}(\lambda_1
x_2+\lambda_2 x_1)}-e^{\mathrm{i}(\lambda_1 x_1+\lambda_2
x_2)}}\Br{e^{\mathrm{i}(\lambda_3 x_2+\lambda_4
x_1)}-e^{\mathrm{i}(\lambda_3 x_1+\lambda_4 x_2)}}}{(x_1-x_2)^4}.
\end{eqnarray*}
Thus $\widehat P^A_{\bdlambda}(\bsx)=\varphi_{\bdlambda}(\bsx)$.
Moreover, its inverse Fourier transform is given by
\begin{eqnarray}
P^A_{\bdlambda}(\bsz)=\frac1{(2\pi)^2}\int_{\bbR^2}e^{-\mathrm{i}\Inner{\bsz}{\bsx}}
\delta(1-z_1-z_2)\widehat
P^A_{\bdlambda}(\bsx)[\dif\bsx]=\frac1{(2\pi)^2}\int_{\bbR^2}e^{-\mathrm{i}\Inner{\bsz}{\bsx}}\varphi_{\bdlambda}(\bsx)[\dif\bsx].
\end{eqnarray}
Since $z_1+z_2=1$, denote $z=z_1$, we get that $P^A_{\bdlambda}(z)$
is supported on $[\lambda_3+\lambda_4,\lambda_1+\lambda_2]$, and the
explicit expression of $P^A_{\bdlambda}(z)$ is obtained by using the
symbolic computation of \textsc{Mathematica}.
\end{proof}

%================================================================%
\section{Qubit-qutrit systems}
\label{sect:5}
%================================================================%

We now consider the case $(m,n)=(2,3)$, so that $N=6$. Let
$\bdlambda=(\lambda_1,\ldots,\lambda_6)\in C_6\cap\Delta_5$ and set
$\Lambda=\diag(\lambda_1,\ldots,\lambda_6)$. The global state is
$\rho_{AB}=\bsU\Lambda\bsU^\dagger$, where $\bsU\sim
(\sfU(6),\mu_{\haar})$ is Haar-distributed, and its marginal states
are
$$
\rho_A=\ptr{B}{\rho_{AB}}\in\rD(\bbC^2)\quad\text{ and }\quad\rho_B=\ptr{A}{\rho_{AB}}\in\rD(\bbC^3).
$$
For the qubit marginal, we use the Bloch representation
$$
\rho_A=\frac12(\I_2+\bsa\cdot\bdsigma), \quad a=\abs{\bsa}\in[0,1].
$$
For the qutrit marginal, let $\bdbeta=(\beta_1,\beta_2,\beta_3)\in
C_3\cap\Delta_2$ whose components are the ordered eigenvalues of
$\rho_B$. We write
$(\lambda_{\max}(\rho_B),\lambda_{\min}(\rho_B))=(\beta_1,\beta_3)=(\beta,\gamma)$.
Then $\beta_2=1-\beta-\gamma$. Thus the full qutrit spectrum is
determined by $(\beta,\gamma)$. The open qutrit Weyl chamber is
$$
\cW_3=\Set{(\beta,\gamma)\in\bbR^2:\beta> 1-\beta-\gamma> \gamma>0},
$$
or equivalently,
$$
\cW_3=\Set{(\beta,\gamma)\in\bbR^2:2\beta+\gamma> 1,
\beta+2\gamma<1,\gamma>0}.
$$
We use the notation $V_6(\bdlambda)=\prod_{1\leqslant i<j\leqslant
6}(\lambda_i-\lambda_j)$. For a $3$-element subset
$J\subset\set{1,\ldots,6}$, let $J^c$ denote its complement and
define
\begin{eqnarray}\label{eq:C2b}
V_J(\bdlambda):=\prod_{i,j\in J: i<j}(\lambda_i-\lambda_j),\quad
w_J(\bdlambda):=2\sum_{j\in J}\lambda_j-1,
\end{eqnarray}
together with the sign $\varepsilon_J:=(-1)^{\sum_{j\in J}j}$. As
before, $(x)_+=\max(x,0)$.

\subsection{Distribution of the qubit Bloch radius}

We first determine the distribution of a fixed Cartesian component
of the qubit Bloch vector. Let
$\alpha=\Tr{\rho_A\sigma_3}=\Tr{(\sigma_3\ot\I_3)\rho_{AB}}$, and
denote its probability density by $q^A_{\bdlambda}(\alpha)$. Its
characteristic function is
$$
\widehat
q^A_{\bdlambda}(s)=\int_{\sfU(6)}e^{\mathrm{i}s\Tr{(\sigma_3\ot\I_3)\bsU\Lambda\bsU^\dagger}}\dif\mu_{\haar}(\bsU).
$$
The eigenvalues of the test matrix $\bsH_s=s\sigma_3\ot\I_3$ are
$$
(s,s,s,-s,-s,-s),
$$
with multiplicities $3$ and $3$. Thus the corresponding HCIZ
integral contains two eigenvalue blocks, each of multiplicity three.
\begin{prop}\label{prop:C2C3}
The characteristic function and probability density of the fixed
Bloch component $\alpha$ are
\begin{eqnarray}\label{eq:C2a0}
\widehat q^A_{\bdlambda}(s) =
\frac{135}{8V_6(\bdlambda)}\sum_{\abs{J}=3}\varepsilon_J
V_J(\bdlambda)V_{J^c}(\bdlambda)\frac{e^{\mathrm{i}sw_J(\bdlambda)}}{(\mathrm{i}s)^9}
\end{eqnarray}
and
\begin{eqnarray}\label{eq:C2a}
q_{\bdlambda}(\alpha) = \frac{135}{8\cdot
8!V_6(\bdlambda)}\sum_{\abs{J}=3}\varepsilon_JV_J(\bdlambda)V_{J^c}(\bdlambda)(w_J(\bdlambda)-\alpha)^8_+.
\end{eqnarray}
The apparent singularity of each summand in Eq.~\eqref{eq:C2a0} at
$s=0$ is removable after taking the complete alternating sum.
\end{prop}

\begin{proof}
By the HCIZ formula,
\begin{eqnarray*}
\widehat q^A_{\bdlambda}(s) = \gamma_6\mathrm{i}^{-15} \lim_{\bsh\to(s,s,s,-s,-s,-s)}\frac{\det\Pa{e^{\mathrm{i}h_i\lambda_j}}^6_{i,j=1}}{V_6(\bsh)V_6(\bdlambda)},
\end{eqnarray*}
where
$$
\gamma_6=\prod^6_{k=1}\Gamma(k)=1!2!3!4!5!=34560.
$$
Apply the confluent identity separately to the two triple-degenerate blocks. For the block tending to $s$, the limiting rows are proportional to $e^{\mathrm{i}s\lambda_j}, \lambda_je^{\mathrm{i}s\lambda_j},\lambda^2_je^{\mathrm{i}s\lambda_j}$, and similarly for the block tending to $-s$. The cross-block Vandermonde contribution is
$$
\prod_{1\leqslant i\leqslant 3<j\leqslant6}(h_i-h_j)\to (2s)^9.
$$
Accounting for the derivative factors and the two confluent signs gives
\begin{eqnarray}\label{eq:limitdetdet}
\lim_{\bsh\to(s,s,s,-s,-s,-s)}\frac{\det\Pa{e^{\mathrm{i}h_i\lambda_j}}^6_{i,j=1}}{V_6(\bsh)} = -\frac1{4(2s)^9}\det\Pa{\begin{array}{c}
                             e^{\mathrm{i}s\lambda_j} \\
                            \lambda_j e^{\mathrm{i}s\lambda_j} \\
                             \lambda^2_j e^{\mathrm{i}s\lambda_j} \\
                             e^{-\mathrm{i}s\lambda_j} \\
                             \lambda_j e^{-\mathrm{i}s\lambda_j} \\
                             \lambda^2_j e^{-\mathrm{i}s\lambda_j}
                           \end{array}
}^6_{j=1}.
\end{eqnarray}
Expanding this determinant along its first three rows gives
\begin{eqnarray}\label{eq:expandingdet}
\det\Pa{\begin{array}{c}
                             e^{\mathrm{i}s\lambda_j} \\
                            \lambda_j e^{\mathrm{i}s\lambda_j} \\
                             \lambda^2_j e^{\mathrm{i}s\lambda_j} \\
                             e^{-\mathrm{i}s\lambda_j} \\
                             \lambda_j e^{-\mathrm{i}s\lambda_j} \\
                             \lambda^2_j e^{-\mathrm{i}s\lambda_j}
                           \end{array}
}^6_{j=1} =
\sum_{\abs{J}=3}\varepsilon_JV_J(\bdlambda)V_{J^c}(\bdlambda)e^{\mathrm{i}s(\Lambda_J-\Lambda_{J^c})},
\end{eqnarray}
where $\Lambda_J:=\sum_{j\in J}\lambda_j$. Since $\Lambda_J+\Lambda_{J^c}=1$, it follows that $\Lambda_J-\Lambda_{J^c}=2\Lambda_J-1=w_J(\bdlambda)$. Substituting Eq.~\eqref{eq:expandingdet} into the HCIZ formula and simplifying the constants gives  Eq.~\eqref{eq:C2a0}.

With the Fourier conventions of Section~\ref{sect:2},
$$
\cF^{-1}\Pa{\frac{e^{\mathrm{i}sw}}{(\mathrm{i}s)^k}}(x) = \frac{(w-x)^k_+}{(k-1)!},
$$
where the denominator is understood with the same distributional
polarization as in the characteristic function. Taking $k=9$ yields
Eq.~\eqref{eq:C2a}.
\end{proof}
We now pass from the fixed-component density to the density. By
local-unitary invariance, the Bloch vector $\bsa$ is rotationally
invariant. Conditional on $a=\abs{\bsa}$, the random variable
$\alpha=a\cos\theta$ is uniformly distributed on $[-a,a]$.
Therefore,
$q_{\bdlambda}(\alpha)=\int^1_{\abs{\alpha}}\frac{p^A_{\bdlambda}(a)}{2a}\dif
a$, where $p^A_{\bdlambda}(a)$ is the density of the Bloch radius.
For $a>0$, $p^A_{\bdlambda}(a)=(-2a)\frac{\dif
q_{\bdlambda}(a)}{\dif a}$.

\begin{thrm}[Qubit Bloch-radius density]
For a random qubit-qutrit state on the regular unitary orbit $\cU_\Lambda$, the probability density of the qubit Bloch radius $a$ is
\begin{eqnarray}\label{eq:p23(a)}
p^A_{\bdlambda}(a) =
\frac{3a}{448V_6(\bdlambda)}\sum_{\abs{J}=3}(-1)^{\sum_{j\in
J}j}V_J(\bdlambda)V_{J^c}(\bdlambda)(w_J(\bdlambda)-a)^7_+,\quad
a\in\bbR_{\geqslant0},
\end{eqnarray}
which is shown in Figure~\ref{fig:PDFBlochr2(3)}. Its support is the
closed interval
\begin{eqnarray}\label{eq:suppp23}
\operatorname{supp}(p^A_{\bdlambda}) =
\Br{0,2\sum^3_{k=1}\lambda_k-1}.
\end{eqnarray}
Consequently, $p^A_{\bdlambda}(a)$  is a univariate piecewise-polynomial density of degree at most eight. Its possible interior breakpoints are contained in
\begin{eqnarray}\label{eq:breakpoints}
\Set{2\sum_{j\in J}\lambda_j-1: \abs{J}=3,\sum_{j\in
J}\lambda_j>\frac12}.
\end{eqnarray}
\end{thrm}
\begin{figure}[h!]\centering
{\begin{minipage}[b]{0.7\linewidth}
\includegraphics[width=1\textwidth]{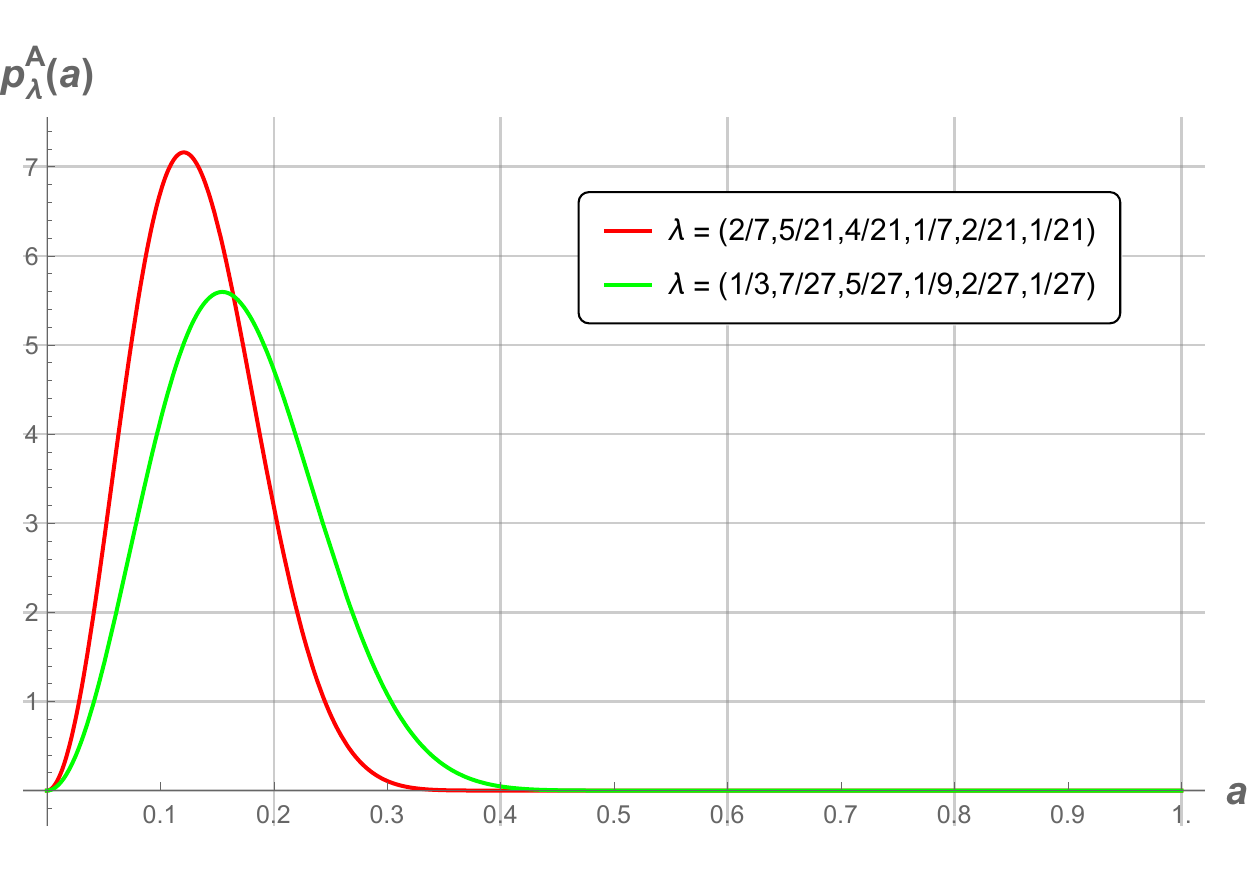}
\end{minipage}}
\caption{The probability density curve of a Bloch radius of a
marginal state of a qubit-qutrit unitary
orbit.}\label{fig:PDFBlochr2(3)}
\end{figure}

\begin{proof}
Differentiating Eq.~\eqref{eq:C2a} gives
$$
\frac{\dif}{\dif \alpha}(w_J-\alpha)^8_+=-8(w_J-\alpha)^7_+.
$$
Hence, using $p^A_{\bdlambda}(a)=(-2a)\frac{\dif
q^A_{\bdlambda}(a)}{\dif a}$,
\begin{eqnarray*}
p^A_{\bdlambda}(a) &=& -2a \frac{\dif}{\dif a}\Br{\frac{135}{8\cdot 8!V_6(\bdlambda)}\sum_{\abs{J}=3}\varepsilon_J V_J(\bdlambda)V_{J^c}(\bdlambda)(w_J-a)^8_+}\\
&=&\frac{3a}{448V_6(\bdlambda)}\sum_{\abs{J}=3}\varepsilon_J V_J(\bdlambda)V_{J^c}(\bdlambda)(w_J-a)^7_+.
\end{eqnarray*}
This proves Eq.~\eqref{eq:p23(a)}.

To determine the upper endpoint of the support, observe that
$$
\lambda_{\max}(\rho_A) = \max_{\norm{\psi}=1}\Tr{(\proj{\psi}\ot\I_3)\rho_{AB}}.
$$
The operator $\proj{\psi}\ot\I_3$ is a rank-three orthogonal projection. Ky Fan's variational principle therefore gives
$$
\lambda_{\max}(\rho_A)\leqslant \lambda_1+\lambda_2+\lambda_3.
$$
Since $\lambda_{\max}(\rho_A)=\frac{1+a}2$, we obtain $a\leqslant 2(\lambda_1+\lambda_2+\lambda_3)-1$. This upper bound is attained by choosing the eigenvectors associated with $\lambda_1,\lambda_2,\lambda_3$ inside the subspace $\ket{0}\ot\bbC^3$, and the remaining eigenvectors inside $\ket{1}\ot\bbC^3$.

The lower endpoint $a=0$ is also attainable. For example, consider the six orthonormal states
$$
\ket{\psi_{k,\pm}} = \frac{\ket{0,k}\pm \ket{1,k+1\pmod3}}{\sqrt{2}},\quad k=0,1,2.
$$
Each of these states has qubit marginal$\frac12\I_2$. Taking them as
the eigenvectors of $\rho_{AB}$, with arbitrary assignment of the
six eigenvalues $\lambda_j$, gives $\rho_A=\frac12\I_2$, and hence
$a=0$.

Finally, $\sfU(6)$ is connected and the map
$$
\bsU\mapsto \abs{\bsa\Pa{\ptr{B}{\bsU\Lambda\bsU^\dagger}}}
$$
is continuous. Its image is therefore a connected compact subset of
$\bbR$ containing both endpoints above, and hence is exactly the
interval in Eq.~\eqref{eq:suppp23}.
\end{proof}

\begin{remark}
Although Eq.~\eqref{eq:suppp23} is written as an alternating sum of
truncated powers, the complete expression is non-negative because it
is the density of a push-forward probability measure. Individual
summands need not be non-negative after multiplication by their
alternating coefficients.

The degree bound is eight, not seven: each active term has the form
$a(w_J-a)^7$, which is a polynomial of degree eight on every chamber
where the set of active truncated powers is fixed.
\end{remark}

\subsection{Joint distribution of the largest and smallest qutrit eigenvalues}

We next determine the spectral distribution of the qutrit marginal.
The derivation has two steps:
\begin{enumerate}
\item calculate the joint distribution of two diagonal entries of
$\rho_B$ in a fixed basis;
\item apply the $\SU(3)$ derivative principle to recover the ordered eigenvalue density.
\end{enumerate}

Let $z_j=\Innerm{j}{\rho_B}{j}$, where $j=1,2,3$. Since
$z_1+z_2+z_3=1$, it suffices to consider the pair $(z_1,z_3)$. Let
$q^B_{\bdlambda}(z_1,z_3)$ denote its joint density, with
characteristic function
\begin{eqnarray}\label{eq:widehatqBst}
\widehat q^B_{\bdlambda}(s,t) =
\int_{\bbR^2}e^{\mathrm{i}(sz_1+tz_3)}q^B_{\bdlambda}(z_1,z_3)\dif
z_1\dif z_3.
\end{eqnarray}
Because
$$
sz_1+tz_3=\Tr{(\I_2\ot\diag(s,0,t))\rho_{AB}},
$$
the relevant test matrix is $\bsH_{s,t}=\I_2\ot\diag(s,0,t)$. Its
eigenvalues are
$$
(s,s,0,0,t,t),
$$
so there are three double-degenerate blocks. Define
\begin{eqnarray}\label{eq:Dlambdast}
D_{\bdlambda}(s,t) = \det\Pa{\begin{array}{c}
                               e^{\mathrm{i}s\lambda_j} \\
                               \lambda_je^{\mathrm{i}s\lambda_j} \\
                               1 \\
                               \lambda_j \\
                               e^{\mathrm{i}t\lambda_j} \\
                               \lambda_je^{\mathrm{i}t\lambda_j}
                             \end{array}
}^6_{j=1}.
\end{eqnarray}
Expanding this determinant gives
\begin{eqnarray}\label{eq:numerator}
D_{\bdlambda}(s,t) = \sum_{\pi\in
S_6}\sign(\pi)\lambda_{\pi(2)}\lambda_{\pi(4)}\lambda_{\pi(6)}
e^{\mathrm{i}s(\lambda_{\pi(1)}+\lambda_{\pi(2)})}e^{\mathrm{i}t(\lambda_{\pi(5)}+\lambda_{\pi(6)})}.
\end{eqnarray}

\begin{prop}[Abelian qutrit characteristic function]
The characteristic function of the pair entries $(z_1,z_3)$ of the
qutrit marginal is
\begin{eqnarray}\label{eq:widehatqBst}
\widehat q^B_{\bdlambda}(s,t) =
-\frac{34560}{V_6(\bdlambda)}\frac{D_{\bdlambda}(s,t)}{s^4t^4(s-t)^4},
\end{eqnarray}
The apparent singularities at $(s,t)=(0,0)$, and $s=t$ are removable
in the complete determinant expression.
\end{prop}

\begin{proof}
Applying the double-confluent identity to each of the three
eigenvalue blocks gives
\begin{eqnarray}
\lim_{\bsh\to(s,s,0,0,t,t)}\frac{\det\Pa{e^{\mathrm{i}h_i\lambda_j}}^6_{i,j=1}}{V_6(\bsh)}=\frac{\mathrm{i}D_{\bdlambda}(s,t)}{s^4t^4(s-t)^4}.
\end{eqnarray}
Indeed, the cross-block part of the Vandermonde tends to
$s^4t^4(s-t)^4$, while the three double-confluent limits supply the
three derivative rows appearing in Eq.~\eqref{eq:Dlambdast}. Since
$\gamma_6\mathrm{i}^{-15}=34560\mathrm{i}$, the HCIZ formula yields
$$
\widehat q^B_{\bdlambda}(s,t)=
\frac{34560\mathrm{i}}{V_6(\bdlambda)}\frac{\mathrm{i}D_{\bdlambda}(s,t)}{s^4t^4(s-t)^4}=
-
\frac{34560}{V_6(\bdlambda)}\frac{D_{\bdlambda}(s,t)}{s^4t^4(s-t)^4}.
$$
This is Eq.~\eqref{eq:widehatqBst}.
\end{proof}
We now apply the $\SU(3)$ derivative principle. On the trace-one
plane, write the ordered eigenvalues as
$\bdbeta=(\beta_1,\beta_2,\beta_3)=(\beta,1-\beta-\gamma,\gamma)$.
The positive-root differential directions restrict to
\begin{eqnarray*}
\partial_{\beta_1}-\partial_{\beta_2}&\longleftrightarrow& \partial_\beta\\
\partial_{\beta_1}-\partial_{\beta_3}&\longleftrightarrow&
\partial_\beta-\partial_\gamma,\\
\partial_{\beta_2}-\partial_{\beta_3}&\longleftrightarrow&
-\partial_\gamma.
\end{eqnarray*}
With the present Vandermonde convention, the resulting Weyl
differential operator is
\begin{eqnarray}\label{eq:derivativeprinciple}
\cL:=\partial_\beta\partial_\gamma(\partial_\beta-\partial_\gamma).
\end{eqnarray}
The positive Weyl denominator is
\begin{eqnarray}\label{eq:posweyldenominator}
V_3(\beta,1-\beta-\gamma,\gamma)=(2\beta+\gamma-1)(\beta-\gamma)(1-\beta-2\gamma).
\end{eqnarray}
The $\SU(3)$ derivative principle gives, in the open Weyl chamber
$\cW_3$, the joint density of the largest and smallest eigenvalues
of $\rho_B$:
\begin{eqnarray}\label{eq:jointpdfpBbetagamma}
p^B_{\bdlambda}(\beta,\gamma)=\frac12V_3(\beta,1-\beta-\gamma,\gamma)\cL
q^B_{\bdlambda}(\beta,\gamma),
\end{eqnarray}
where the factor $\frac1{1!2!}=\frac12$ is the $\SU(3)$
normalization in the derivative principle.

Under inverse Fourier transformation, the operator $\cL$ corresponds
to multiplication by
$$
(-\mathrm{i}s)(-\mathrm{i}t)(-\mathrm{i}(s-t))=\mathrm{i}st(s-t).
$$
Combining this multiplier with Eq.~\eqref{eq:widehatqBst}, we obtain
\begin{eqnarray}\label{eq:LqBbetagamma}
\cL q^B_{\bdlambda}(\beta,\gamma)=\frac{34560}{V_6(\bdlambda)}
\cF^{-1}\Br{\frac{D_{\bdlambda}(s,t)}{(\mathrm{i}s)^3(\mathrm{i}t)^3(\mathrm{i}(s-t))^3}}(\beta,\gamma).
\end{eqnarray}
To evaluate the inverse transform, define the bivariate
truncated-power function
\begin{eqnarray}\label{eq:Tintegral}
\cT(x,y)=\frac18\int^x_{\max(0,-y)}(x-r)^2(y+r)^2r^2\dif r,
\end{eqnarray}
with the convention that $\cT(x,y)=0$ whenever $x<\max(0,-y)$.
Equivalently,
\begin{eqnarray}\label{eq:cTexpression}
\cT(x,y)=
\begin{cases}
\frac{x^5(2x^2+7xy+7y^2)}{1680},&\text{if
}x\geqslant0,y\geqslant0,\\
\frac{(x+y)^5(2x^2-3xy+2y^2)}{1680},&\text{if }x\geqslant0,
-x\leqslant
y\leqslant0,\\
0,&\text{otherwise}.
\end{cases}
\end{eqnarray}
Its derivation is presented in Appendix~\ref{app:Txy}. Thus $\cT$ is
supported on the closed cone
$$
\Set{(x,y)\in\bbR^2: x\geqslant0,x+y\geqslant0}
$$
and is piecewise polynomial of degree seven.

The basic inverse-transform identity is
\begin{eqnarray}\label{eq:expUV}
\cF^{-1}\Br{\frac{e^{\mathrm{i}(sU+tV)}}{(\mathrm{i}s)^3(\mathrm{i}t)^3(\mathrm{i}(s-t))^3}}(\beta,\gamma)=\cT(U-\beta,V-\gamma).
\end{eqnarray}
Indeed,
$$
\frac1{(\mathrm{i}s)^3} =
\int^\infty_0\frac{A^2}{2!}e^{-\mathrm{i}sA}\dif A,
$$
and analogously for the other two factors. Consequently,
\begin{eqnarray}
&&\cF^{-1}\Br{\frac{e^{\mathrm{i}(sU+tV)}}{(\mathrm{i}s)^3(\mathrm{i}t)^3(\mathrm{i}(s-t))^3}}(\beta,\gamma)\\
&&=\frac1{(2!)^3}\int_{\bbR^3_{\geqslant0}}A^2B^2C^2\delta(U-\beta-A-C)\delta(V-\gamma-B+C)\dif
A\dif B\dif C.
\end{eqnarray}
Writing $(x,y)=(U-\beta,V-\gamma)$, the delta functions give
$(A,B)=(x- C,y+C)$. The non-negativity constraints become
$\max(0,-y)\leqslant C\leqslant x$, which gives
Eq.~\eqref{eq:Tintegral}.

\begin{thrm}[Joint density of largest and smallest qutrit eigenvalues]
For a random qubit-qutrit state on the regular unitary orbit
$\cU_{\Lambda}$, the joint probability density of largest and
smallest eigenvalues,
$$
(\beta,\gamma)=(\lambda_{\max}(\rho_B),\lambda_{\min}(\rho_B))
$$
of the qutrit marginal state
$\rho_B=\ptr{A}{\bsU\Lambda\bsU^\dagger}$, is
\begin{eqnarray}\label{eq:qutritpBbetagamma}
p^B_{\bdlambda}(\beta,\gamma)
&=&\frac{17280V_3(\beta,1-\beta-\gamma,\gamma)}{V_6(\bdlambda)}\sum_{\pi\in
S_6}\sign(\pi)\lambda_{\pi(2)}\lambda_{\pi(4)}\lambda_{\pi(6)}\notag\\
&&\times
\cT(\lambda_{\pi(1)}+\lambda_{\pi(2)}-\beta,\lambda_{\pi(5)}+\lambda_{\pi(6)}-\gamma),
\end{eqnarray}
for
$(\beta,\gamma)\in\cW_3=\Set{(\beta,\gamma):\beta>1-\beta-\gamma>\gamma\geqslant0}$,
and is zero outside the closed qutrit Weyl chamber
$\overline{\cW}_3=\Set{(\beta,\gamma):\beta\geqslant1-\beta-\gamma\geqslant\gamma\geqslant0}$.
Its support is the one-marginal spectral polytope
\begin{eqnarray}\label{eq:support}
\op{supp}(p^B_{\bdlambda}) &=&
\overline{\Set{(\beta,\gamma)\in\cW_3:
p^B_{\bdlambda}(\beta,\gamma)>0}}\notag\\
&=&\Set{(\lambda_{\max}(\ptr{A}{\bsU\Lambda\bsU^\dagger}),\lambda_{\min}(\ptr{A}{\bsU\Lambda\bsU^\dagger})):\bsU\in\sfU(6)}
\end{eqnarray}
In particular, every point in the support satisfies the necessary
bounds
\begin{eqnarray}\label{eq:beta3gamma}
\lambda_5+\lambda_6\leqslant \gamma\leqslant \frac13\leqslant
\beta\leqslant \lambda_1+\lambda_2.
\end{eqnarray}
\end{thrm}
The graph of $p^B_{\bdlambda}(\beta,\gamma)$ in
Eq.~\eqref{eq:qutritpBbetagamma} is depicted in
Figure~\ref{fig:jointdensity(2)3}.
\begin{figure}[h!]\centering
{\begin{minipage}[b]{0.8\linewidth}
\includegraphics[width=1\textwidth]{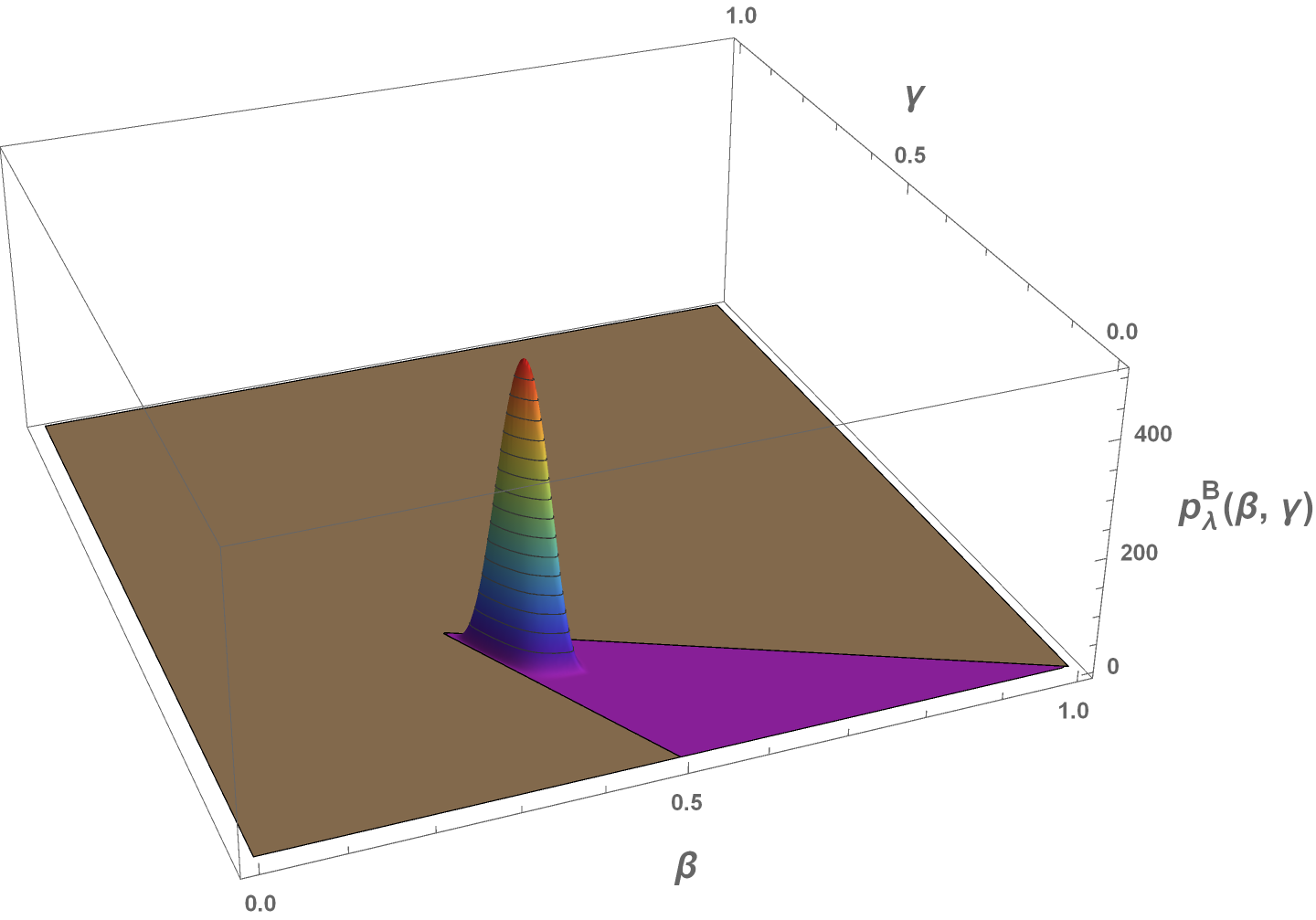}
\end{minipage}}
\caption{The joint density of largest and smallest eigenvalues of a
qutrit marginal state of a qubit-qutrit unitary orbit, where
$\bdlambda=(0.2380,0.2342,0.2296,0.1136,0.1109,0.0736)$. Gray color
is used outside the support.}\label{fig:jointdensity(2)3}
\end{figure}

\begin{proof}
Substituting the determinant expansion Eq.~\eqref{eq:numerator} into
Eq.~\eqref{eq:LqBbetagamma}, and then applying Eq.~\eqref{eq:expUV}
gives
\begin{eqnarray}
\cL q^B_{\bdlambda}(\beta,\gamma) &=&
\frac{34560}{V_6(\bdlambda)}\sum_{\pi\in
S_6}\sign(\pi)\lambda_{\pi(2)}\lambda_{\pi(4)}\lambda_{\pi(6)}\\
&&\times
\cT(\lambda_{\pi(1)}+\lambda_{\pi(2)}-\beta,\lambda_{\pi(5)}+\lambda_{\pi(6)}-\gamma)
\end{eqnarray}
Applying Eq.~\eqref{eq:jointpdfpBbetagamma} proves Eq.~\eqref{eq:qutritpBbetagamma}.

The equality in Eq.~\eqref{eq:support} follows from the definition of the push-forward measure. Haar measure has full support on $\sfU(6)$, and the map
$$
\bsU\mapsto (\lambda_{\max}(\ptr{A}{\bsU\Lambda\bsU^\dagger}),
\lambda_{\min}(\ptr{A}{\bsU\Lambda\bsU^\dagger}))
$$
is continuous. Therefore, the support of its push-forward is precisely its compact image.

To obtain the bounds in Eq.~\eqref{eq:beta3gamma}, for any unit
vector $\ket{\varphi}\in\bbC^3$, consider the rank-two projection
$\bsP_{\varphi}=\I_2\ot\proj{\varphi}$. Then
$\Innerm{\varphi}{\rho_B}{\varphi}=\Tr{\rho_{AB}\bsP_{\varphi}}$. Ky
Fan's maximal and minimal variational principles give
$$
\lambda_5+\lambda_6\leqslant \Tr{\rho_{AB}\bsP_{\varphi}}\leqslant
\lambda_1+\lambda_2.
$$
Taking the minimum and maximum over $\ket{\varphi}$ yields
$$
\lambda_1+\lambda_2\geqslant \beta\geqslant\gamma\geqslant \lambda_5+\lambda_6.
$$
The remaining inequalities $\beta\geqslant \frac13\geqslant\gamma$ follow from the ordering and unit-trace condition for a qutrit spectrum. These bounds are necessary but, in general, need not constitute a complete irredundant description of the marginal polytope.
\end{proof}
The complete constraints on the spectra
$(\bdlambda(\rho_{AB}),\bdlambda(\rho_A),\bdlambda(\rho_B))$ were
obtained in \cite{Klyachko2004}. In contrast, the support of
$p^B_{\bdlambda}(\beta,\gamma)$ defined above is distinct from
Klyachko's inequalities for the qubit-qutrit system.

\begin{remark}[Reduction from $720$ to $90$ terms]
The sum over $S_6$ in Eq.~\eqref{eq:qutritpBbetagamma} can be
reduced from $6!=720$ terms to $90$ terms. Let
$$
P_{2,2,2}=\Set{(I,J,K):I\sqcup J\sqcup
K=\set{1,\ldots,6},\abs{I}=\abs{J}=\abs{K}=2},
$$
where $(I,J,K)$  is an ordered triple of disjoint pairs. If
$$
I=\set{i_1<i<_2}, \quad J=\set{j_1<j_2}, \quad K=\Set{k_1<k_2},
$$
define $\delta_I=\lambda_{i_2}-\lambda_{i_1},
\Lambda_I=\lambda_{i_1}+\lambda_{i_2}$, and analogously for $J$ and
$K$. Also define
$$
\varepsilon(I,J,K)=\sign(i_1,i_2,j_1,j_2,k_1,k_2).
$$
Then the determinant in Eq.~\eqref{eq:Dlambdast} can be written as
$$
D_{\bdlambda}(s,t)=\sum_{(I,J,K)\in
P_{2,2,2}}\varepsilon(I,J,K)\delta_I\delta_J\delta_Ke^{\mathrm{i}(s\Lambda_I
+t\Lambda_K)}.
$$
Consequently,
\begin{eqnarray}\label{eq:pBlbg}
p^B_{\bdlambda}(\beta,\gamma)=\frac{17280V_3(\beta,1-\beta-\gamma,\gamma)}{V_6(\bdlambda)}\sum_{(I,J,K)\in
P_{2,2,2}}\varepsilon(I,J,K)\delta_I\delta_J\delta_K\cT(\Lambda_I-\beta,\Lambda_K-\gamma).
\end{eqnarray}
Because $\abs{P_{2,2,2}}=\frac{6!}{(2!)^3}=90$, Eq.~\eqref{eq:pBlbg}
is substantially more efficient for symbolic and numerical
evaluation.
\end{remark}

\begin{remark}[Piecewise-polynomial structure]
The truncated-power function $\cT$ is piecewise polynomial of degree
$7$. The qutrit Weyl denominator
$V_3(\beta,1-\beta-\gamma,\gamma)=(2\beta+\gamma-1)(\beta-\gamma)(1-\beta-2\gamma)$
is a polynomial of degree $3$. Therefore
$p^B_{\bdlambda}(\beta,\gamma)$ is piecewise polynomial of degree at
most $10$.

The walls of the chamber decomposition are contained in
$$
\beta=\Lambda_I,\quad\gamma=\Lambda_K,\quad\beta+\gamma=\Lambda_I+\Lambda_K,
$$
where $I$ and $K$ are disjoint two-element subsets of
$\set{1,\ldots,6}$.These walls arise from the three boundary rays
$x=0,y=0,x+y=0$ of the bivariate truncated-power function
$\cT(x,y)$.
\end{remark}

%================================================================%
\section{Concluding remarks}
\label{sect:6}
%================================================================%

The marginal Bloch radii are natural invariants that encode local
purity, constrain the possible marginal states, and provide a
practical measure of how global correlations are distributed between
local and nonlocal degrees of freedom within a fixed unitary orbit.
Motivated by this, we studied marginal spectral distributions
induced by the Haar orbital measure on a bipartite unitary orbit
with fixed regular spectrum. The joint matrix-valued characteristic
function is an HCIZ integral whose external eigenvalues are the
pairwise sums $x_i+y_j$. In low dimensions, confluent limits reduce
this integral to rational Fourier transforms, and their inverses are
naturally expressed by univariate and multivariate truncated-power
functions.

In the two-qubit case, we obtained the full joint density of the two
marginal Bloch radii. Its support is the Bravyi compatibility
region, while the density supplies a probabilistic refinement of
that deterministic region. We also derived a compact one-variable
formula for either individual marginal spectrum. In the qubit-qutrit
case, we obtained the qubit one-marginal Bloch-radius density and
the qutrit one-marginal density of its largest and smallest
eigenvalues; these are separate one-marginal laws, and deriving
their full joint distribution remains open.

The basic truncated-power functions appearing in the formulas are
supported on polyhedral cones. Compact support arises only after
taking the complete alternating sums dictated by the HCIZ
determinant. The resulting densities are piecewise polynomial on
chambers determined by subset sums of the fixed global spectrum, in
agreement with the Duistermaat-Heckman description of projected
coadjoint-orbit measures.

Several directions can be considered in the future research.
\begin{itemize}
\item \textbf{Full joint distributions and higher-dimensional systems.}
We continue to consider marginal states on a unitary orbit of fixed
spectrum $\Lambda$. The natural next step is the joint density
$p_{\bdlambda}(a,\beta,\gamma)$ of the qubit Bloch radius $a$ and
the largest $\beta$ and smallest qutrit eigenvalues $\gamma$ in the
qubit-qutrit case, which should be a multivariate spline supported
on the full qubit-qutrit marginal polytope. For general $m$ and $n$,
a systematic treatment will require efficient confluent HCIZ
formulas with general block multiplicities, possibly via divided
differences, Schur-function expansions, residue methods, or
multivariate truncated powers.
\item \textbf{Explicit descriptions of marginal polytopes.}
Marginal spectra alone do not determine whether a mixed bipartite
state is separable or entangled. Therefore, although the spectral
densities derived in this work fully describe the random
distribution of the marginal spectra, integrating them over a subset
of the marginal polytope cannot directly yield a general
separability probability. Nevertheless, these densities are useful
for studying quantities that are determined or constrained by the
local spectra, such as marginal purities (e.g., $\bbE[(1+a^2)/2]$),
local von Neumann entropies, and spectral asymmetry. They can also
serve as a building block in more refined calculations, for instance
by combining them with distributions of correlation tensors to
investigate conditional distributions of entanglement measures given
fixed marginal spectra. The support of each density is exactly the
corresponding spectral marginal polytope. In the two-qubit case,
this polytope is completely described by the Bravyi-Klyachko
inequalities, whereas in higher dimensions it is generally only
characterized abstractly via Klyachko's representation-theoretic
criteria. Thus, a valuable future direction would be to
systematically map the walls appearing in the spline formulas (which
originate from subset sums of the global eigenvalues) to a minimal
generating set of Klyachko-type inequalities. In particular, for the
qutrit marginal in the qubit-qutrit system, the elementary Ky Fan
bounds derived here provide only necessary conditions and may not be
sufficient; finding a concise and complete inequality system for
this marginal polytope remains an important open problem.
\item \textbf{Degenerate global spectra.}
When some $\lambda_j$ coincide, the apparent singularity from $V_N(\bdlambda)=0$ is removable in the complete orbital integral,
and the degenerate case can be obtained by continuous confluent limits.
Such limits may simplify the final formulas and are relevant for global states with few distinct eigenvalues,
including normalized projectors and depolarized pure states.
\item \textbf{Efficient computation and wall-crossing.}
Direct numerical evaluation can suffer from cancelation between large alternating terms,
especially near chamber walls or when global eigenvalues are close.
Stable implementations should group terms and exploit the chamber structure of the truncated powers.
Wall-crossing formulas for Duistermaat-Heckman measures may offer an alternative
by propagating a polynomial piece across adjacent walls instead of
evaluating the full alternating sum independently on every chamber.
\end{itemize}

In conclusion, the present framework connects three complementary
aspects of the quantum marginal problem: HCIZ characteristic
functions, truncated powers and splines, marginal spectral
distributions. The support of each density recovers the associated
deterministic compatibility region, while the density itself
describes how Haar orbital measure is distributed inside that
region. The explicit two-qubit and qubit-qutrit formulas illustrate
how random-matrix methods, Fourier analysis, and symplectic geometry
can be combined to obtain quantitative probabilistic information
beyond spectral compatibility alone.

\subsection*{Acknowledgments}

The author gratefully acknowledges Dr. Jiyuan Zhang for insightful
discussions on the treatment of confluent
Harish-Chandra-Itzykson-Zuber integrals used in this work.

\subsection*{Declaration on the use of generative AI}

The author designed the study, carried out the mathematical
analysis, and prepared the initial manuscript. Generative-AI tools
were used for language refinement and assistance with computational
checks. All mathematical statements and computations were
independently reviewed and verified by the author, who takes full
responsibility for the content of the manuscript.

%------------------------------------------------------------%
\newpage
\appendix
\appendixpage
\addappheadtotoc
%------------------------------------------------------------%

%=====================================================================%
\section{Introduction to truncated power function}\label{app:tpf}
%=====================================================================%

In mathematics, particularly in approximation theory, numerical
analysis, and signal processing, a \emph{truncated power function}
(often denoted with a subscript plus sign, e.g., $(x-a)^n_+$) is a
piecewise-defined function that acts as a "switch"---it is
identically zero before a cutoff point $a$, and behaves like a
standard power function $(x-a)^n$ afterward.

It serves as the fundamental building block for polynomial splines
and is deeply connected to distribution theory (as the $n$-fold
integral of the Dirac delta function).
\begin{definition}
For a non-negative integer $n\in\bbZ_{\geqslant0}$ and a real-valued
cutoff point $a$, the truncated power function is defined as:
\begin{eqnarray*}
(x - a)^n_+ :=
\begin{cases}
(x - a)^n, & x \geqslant a, \\
0, & x < a.
\end{cases}
\end{eqnarray*}
\end{definition}
(Note: At $x = a$, the value is generally defined as $0$ for $n=0$
to make it right-continuous, though the exact value at a single
point is often irrelevant in integration and interpolation.)

Special cases are included here:
\begin{itemize}
\item $n = 0$: This is the \emph{Heaviside step function} shifted to
$a$:
$$
(x - a)^0_+ = \begin{cases} 1, & x \geqslant a \\ 0, & x < a
\end{cases}
$$
\item $n = 1$: This is the \emph{ramp function}:
$$
(x - a)^1_+ = \max(x - a, 0).
$$
\item $n = 2$: This is the \emph{quadratic hinge function}, widely used
in machine learning (e.g., Support Vector Machines) and structural
engineering.
\end{itemize}
The function $(x-a)^n_+$ is exactly $n-1$ times continuously
differentiable. At $a$, the $n$-th derivative has a jump
discontinuity. The following identity will be used
\begin{eqnarray*}
\frac{\dif}{\dif x}(x-a)^n_+ = n(x-a)^{n-1}_+,\quad\forall
n\geqslant1.
\end{eqnarray*}

%=====================================================================%
\section{The relationship between $p(a)$ and $q(\alpha)$}
\label{app:pvsq}
%=====================================================================%

Let $\bsa=(a_1,a_2,a_3)\in \bbR^3$ be the Bloch vector. Rotational
invariance means its density depends only on
$a=\abs{\bsa}=\sqrt{a_1^2+a_2^2+a_3^2}$. So write the joint density
as $f(\bsa)=g(a)$.

Let $p(a)$ be the density of the radial variable $a$. Since the
surface area of a sphere of radius $a$ is $4\pi a^2$, it follows
that $p(a)=4\pi a^2 g(a)$. Now let $\alpha=a_3$. The density of
$\alpha$ is
\begin{eqnarray*}
q(\alpha)=\int_{a_1^2+a_2^2\leqslant 1-\alpha^2}
g\Pa{\sqrt{a_1^2+a_2^2+\alpha^2}}\dif a_1\dif a_2.
\end{eqnarray*}
Using polar coordinates in the $(a_1,a_2)$-plane, with
$r=\sqrt{a_1^2+a_2^2}$,
\begin{eqnarray*}
q(\alpha)=2\pi\int_0^{\sqrt{1-\alpha^2}}
rg\Pa{\sqrt{r^2+\alpha^2}}\dif r.
\end{eqnarray*}
Make the change of variables $a=\sqrt{r^2+\alpha^2}$. Then $a\dif
a=r \dif r$ and when $r=0, a=\abs{\alpha}$; when
$r=\sqrt{1-\alpha^2},a=1$. Hence
\begin{eqnarray*}
q(\alpha)=2\pi\int^1_{\abs{\alpha}} a g(a)\dif
a=\int^1_{\abs{\alpha}}\frac{p(a)}{2a}\dif a.
\end{eqnarray*}
where in the last equality, we used the fact that $p(a)=4\pi
a^2g(a)$. Therefore now take the derivative at $\alpha>0$. Then
\begin{eqnarray*}
q(\alpha)=\int^1_\alpha\frac{p(a)}{2a}\dif a.
\end{eqnarray*}
Differentiating with respect to $\alpha=a$,
$q'(a)=-\frac{p(a)}{2a}$, i.e.,$p(a)=(-2a)q'(a)$ for $a>0$.\qed

%=====================================================================%
\section{The relationship between $p(a,b)$ and $q(\alpha,\beta)$}
\label{app:pabvsqab}
%=====================================================================%

To prove the two-dimensional relations, we follow the exact same
logic used for the single-vector case, but applied to the joint
distribution of two Bloch vectors that need not be probabilistically
independent. The relevant property is invariance under independent
rotations: $(\bsa,\bsb)\mapsto (\bsR_A\bsa,\bsR_B\bsb)$ for all
$\bsR_A,\bsR_B\in\SO(3)$.

Let the two Bloch vectors be $\bsa,\bsb\in\bbR^3$. Because of
rotational invariance, their joint density depends only on their
lengths: $a=\abs{\bsa}$ and $b=\abs{\bsb}$. Denote this joint
density over the six-dimensional space as $g(a,b)$.
\begin{itemize}
\item\textbf{Step 1: Relate the radial joint density $p(a,b)$ to
$g(a,b)$.} For a single vector of length $r$, the spherical volume
element is $r^2\sin\theta\dif r\dif\theta\dif\phi$. Integrating over
the angular coordinates $(\theta,\phi)$ gives the factor $4\pi$. For
two vectors, we have independent angular integrations, yielding
$(4\pi)^2=16\pi^2$. The volume element in $\bbR^6$ is
$$
a^2b^2\sin\theta_a\sin\theta_b\dif a\dif
b\dif\theta_a\dif\phi_a\dif\theta_b\dif\phi_b.
$$
Thus, the joint density of the radii $(a,b)$ is obtained by
integrating out all angular variables:
\begin{eqnarray*}
p(a,b)=16\pi^2 a^2b^2g(a,b)\Longleftrightarrow
g(a,b)=\frac{p(a,b)}{16\pi^2a^2b^2}.
\end{eqnarray*}
\item\textbf{Step 2: Express $q(\alpha,\beta)$ as an integral over transverse
components.} Let $\alpha=a_3$ and $\beta=b_3$ be the fixed Cartesian
components. The marginal density of $(\alpha,\beta)$ is
$$
q(\alpha,\beta) =
\int_{\bbR^2\times\bbR^2}g\Pa{\sqrt{r^2_a+\alpha^2},\sqrt{r^2_b+\beta^2}}\dif
a_1\dif a_2\dif b_1\dif b_2,
$$
where $r_a=\sqrt{a^2_1+a^2_2}$ and $r_b=\sqrt{b^2_1+b^2_2}$. Using
polar coordinates in each transverse plane, $\dif a_1\dif
a_2=r_a\dif r_a\dif\phi_a$ and $\dif b_1\dif b_2=r_b\dif
r_b\dif\phi_b$. Integrating over the angles gives $(2\pi)^2=4\pi^2$.
Hence
\begin{eqnarray*}
q(\alpha,\beta)=4\pi^2\int^{\sqrt{1-\alpha^2}}_0\int^{\sqrt{1-\beta^2}}_0
r_ar_b g\Pa{\sqrt{r^2_a+\alpha^2,r^2_b+\beta^2}}\dif r_a\dif r_b.
\end{eqnarray*}
Now change variables: $a=\sqrt{r^2_a+\alpha^2}$ and
$b=\sqrt{r^2_b+\beta^2}$. Then $a\dif a=r_a\dif r_a$ and $r_b\dif
r_b=b\dif b$. The limits become: when $r_a=0,a=\abs{\alpha}$; when
$r_a=\sqrt{1-\alpha^2},a=1$. Similarly for $b$. Therefore,
\begin{eqnarray*}
q(\alpha,\beta)&=&4\pi^2 \int^1_{\abs{\alpha}}\int^1_{\abs{\beta}}
ab
g(a,b)\dif a\dif b\\
&=& \int^1_{\abs{\alpha}}\int^1_{\abs{\beta}}\frac{p(a,b)}{4ab}\dif
a\dif b,
\end{eqnarray*}
where we used the fact that $g(a,b)=\frac{p(a,b)}{16\pi^2a^2b^2}$.
\item\textbf{Step 3: Derive the inverse relation.} For
$(\alpha,\beta)\in\bbR^2_{>0}$, Differentiate with respect to
$\alpha=a$ and then $\beta=b$, we get that
$p(a,b)=(4ab)\partial_a\partial_b q(a,b)$ for $(a,b)\in\bbR^2_{>0}$.
\end{itemize}
Thus, the relation is rigorously proven by the change of variables
and differentiation of the integral form.\qed

%====================================================================================%
\section{Derivation and evaluation of the integral in Eq.~\eqref{eq:Gxy}}
\label{app:computing}
%====================================================================================%

To this end, we introduce the following four vectors:
$$
\bsw_1:=(1,0)^\t,\bsw_2:=(0,1)^\t,\bsw_3:=\bsw_1+\bsw_2,\bsw_4:=\bsw_1-\bsw_2.
$$
Apparently
$$
\bbR^2=\op{Span}_{\bbR}\set{\bsw_1,\bsw_2,\bsw_3,\bsw_4}.
$$
Denote $\bsW=(\bsw_1,\bsw_2,\bsw_3,\bsw_4)$, which is a $2\times 4$
matrix. Using delta function of vector argument \cite{Zhang2021},
consider the following \emph{convex polytope} in
$\bbR^4_{\geqslant0}$ (parameterized by
$\bdomega=(x,y)^\t\in\bbR^2$), defined by
\begin{eqnarray*}
P(\bdomega) :=
\Set{\bst=(t_1,t_2,t_3,t_4)^\t\in\bbR^4_{\geqslant0}\mid
\bsW\bst=\sum^4_{k=1}t_k\bsw_k=\bdomega}.
\end{eqnarray*}
Its Lebesgue volume is given by
\begin{eqnarray*}
\vol_L[P(\bdomega)] =
\sqrt{\det(\bsW\bsW^\t)}\int_{\bbR^4_{\geqslant0}}\delta\Pa{\bdomega-\bsW\bst}[\dif
\bst] = 3\int_{\bbR^4_{\geqslant0}}\delta\Pa{\bdomega-\bsW\bst}[\dif
\bst].
\end{eqnarray*}
Its Fourier transform $\bdomega=(x,y)^\t\to \bsz=(s,t)^\t$ is given
\begin{eqnarray*}
&&\int_{\bbR^2}\vol_L[P(\bdomega)]e^{\mathrm{i}\Inner{\bdomega}{\bsz}}[\dif\bdomega]
= 3
\int_{\bbR^2}[\dif\bdomega]e^{\mathrm{i}\Inner{\bdomega}{\bsz}}\int_{\bbR^4_{\geqslant0}}\delta\Pa{\bdomega-\bsW\bst}[\dif
\bst]\\
&&=3\int_{\bbR^4_{\geqslant0}}[\dif
\bst]\int_{\bbR^2}[\dif\bdomega]e^{\mathrm{i}\Inner{\bdomega}{\bsz}}\delta\Pa{\bdomega-\bsW\bst}=3\int_{\bbR^4_{\geqslant0}}[\dif
\bst]e^{\mathrm{i}\Inner{\bsW\bst}{\bsz}}\\
&&=\frac{3}{st(s+t)(s-t)}.
\end{eqnarray*}
This indicates that
$$
\cF^{-1}\Pa{\frac1{st(s+t)(s-t)}}(x,y) = \frac13\vol_L[P(\bdomega)]
=\int_{\bbR^4_{\geqslant0}}\delta\Pa{\bdomega-\bsW\bst}[\dif \bst].
$$
In summary,
\begin{eqnarray*}
G(x,y)=G(\bdomega)=\int_{\bbR^4_{\geqslant0}}\delta\Pa{\bdomega-\bsW\bst}[\dif
\bst].
\end{eqnarray*}
Let us calculate the above integral. Note that $\bdomega=\bsW\bst$
means that
$$
\begin{cases}
x=t_1+t_3+t_4,\\
y=t_2+t_3-t_4.
\end{cases}
$$
Eliminate $t_1,t_2$ from the delta functions: $t_1=x-(t_3+t_4)$ and
$t_2=y-(t_3-t_4)$. The nonnegativity conditions
$\bst\in\bbR^4_{\geqslant0}$ become
\begin{eqnarray*}
\begin{cases}
t_3\geqslant0,\\
t_4\geqslant0,\\
t_3+t_4\leqslant x,\\
t_3-t_4\leqslant y.
\end{cases}
\end{eqnarray*}
Hence
\begin{eqnarray*}
G(x,y)= \op{Area}\Set{(t_3,t_4)\in\bbR^2_{\geqslant0}\mid
t_3+t_4\leqslant x,t_3-t_4\leqslant y}.
\end{eqnarray*}
In order to calculate it, we perform a change of variables
$(t_3,t_4)\to (\zeta,\eta)$ via $\zeta=t_3+t_4$ and $\eta=t_3-t_4$.
Its Jacobian is calculated as
$$
\Abs{\det\Pa{\frac{\partial(\zeta,\eta)}{\partial(t_3,t_4)}}}= 2.
$$
Then $\dif \zeta\dif \eta=2\dif t_3\dif t_4$ or $\dif t_3\dif
t_4=\frac12\dif\zeta\dif\eta$. Moreover, via the above change of
variables, the region $\Set{(t_3,t_4)\in\bbR^2_{\geqslant0}\mid
t_3+t_4\leqslant x,t_3-t_4\leqslant y}$ is transformed into the
following form:
\begin{eqnarray*}
&&\Set{(\zeta,\eta)\in\bbR^2\mid\tfrac{\zeta+\eta}2\geqslant0,\tfrac{\zeta-\eta}2\geqslant0,
\zeta\leqslant x, \eta\leqslant y} \\
&&=\Set{(\zeta,\eta)\in\bbR^2\mid 0\leqslant\zeta\leqslant x,
-\zeta\leqslant \eta\leqslant \zeta,\eta\leqslant y}=:\Omega_{x,y}.
\end{eqnarray*}
Based on this observation, if $y\geqslant -\zeta$, we get that
$$
\Omega_{x,y}=\Set{(\zeta,\eta)\in\bbR^2\mid 0\leqslant\zeta\leqslant
x, -\zeta\leqslant \eta\leqslant \min(\zeta,y)}
$$
and thus
\begin{eqnarray*}
G(x,y)&=&\frac12\int_{\Omega_{x,y}} \dif\zeta\dif\eta
=\frac12\int^x_0\dif\zeta\int^{\min(\zeta,y)}_{-\zeta}\dif\eta\\
&=&\frac12\int^x_0[\min(\zeta,y)+\zeta]\dif \zeta.
\end{eqnarray*}
If $y<-\zeta$, then $\Omega_{x,y}=\emptyset$, $G(x,y)=0$. In a word,
we can write
\begin{eqnarray*}
G(x,y)=\frac12\int^x_0 (\min(\zeta,y)+\zeta)_+\dif\zeta.
\end{eqnarray*}
Evaluating this integral in the relevant regions gives the desired
expression, that is, Eq.~\eqref{eq:Gxy}. Next we compute the
integral
\begin{eqnarray*}
\cI(x,y):=\frac12\int^x_0 (\min(\zeta,y)+\zeta)_+\dif \zeta, \quad
(x,y)\in\bbR^2.
\end{eqnarray*}
First, note the integrand $g(\zeta):=(\min(\zeta,y)+\zeta)_+$.
\begin{itemize}
\item If $\zeta\leqslant y$, then $\min(\zeta,y)=\zeta$, so $g(\zeta)=(2\zeta)_+$.
\item If $\zeta\geqslant y$, then $\min(\zeta,y)=y$, so $g(\zeta)=(\zeta+y)_+$.
\end{itemize}
The zeros of the positive part:
\begin{itemize}
\item If $y\geqslant 0$, then
$g(\zeta)>0\iff \zeta>0$, and
\begin{eqnarray*}
g(\zeta)=
\begin{cases}
2\zeta,&0<\zeta\leqslant y,\\
\zeta+y,&\zeta>y.
\end{cases}
\end{eqnarray*}
\item If $y<0$, then $g(\zeta)>0\iff \zeta>-y$, and since then $\zeta>y$
automatically, we have
\begin{eqnarray*}
g(\zeta)=\zeta+y,\quad \zeta>-y.
\end{eqnarray*}
Thus the integral depends on the relative positions of $x$ and $y$.
\end{itemize}
\textbf{Case 1: $y\geqslant 0$.}
\begin{itemize}
\item If $x\leqslant 0$, the integration interval $[0,x]$ (or $[x,0]$)
lies in the non-positive region, where the integrand is $0$, so
$\cI(x,y)=0$.
\item If $0\leqslant x\leqslant y$, then
\begin{eqnarray*}
\int^x_0 g(\zeta)\dif \zeta=\int^x_0 2\zeta \dif \zeta=x^2,
\end{eqnarray*}
hence $\cI(x,y)=\frac{x^2}{2}$.
\item If $x\geqslant y$, then
\begin{eqnarray*}
\int^x_0 g(\zeta)\dif \zeta=\int^y_0 2\zeta \dif \zeta+\int^x_y
(\zeta+y)\dif \zeta =y^2+\Pa{\frac{(x+y)^2}{2}-2y^2}
=\frac{(x+y)^2}{2}-y^2,
\end{eqnarray*}
so
$\cI(x,y)=\frac{1}{2}\Pa{\frac{(x+y)^2}{2}-y^2}=\frac{x^2+2xy-y^2}{4}$.
\end{itemize}
\textbf{Case 2: $y<0$.} Here $-y>0$, and $g(\zeta)=0$ for
$\zeta\leqslant -y$, while for $\zeta>-y$ we have
$g(\zeta)=\zeta+y$.
\begin{itemize}
\item If $x\leqslant -y$, the interval does not include the non-zero region,
so $\cI(x,y)=0$.
\item If $x\geqslant -y$, then
\begin{eqnarray*}
\int^x_0 g(\zeta)\dif \zeta=\int^x_{-y} (\zeta+y)\dif \zeta
=\frac{(x+y)^2}{2},
\end{eqnarray*}
hence $\cI(x,y)=\frac{(x+y)^2}{4}$.
\end{itemize}
At the boundaries (e.g. $x=y, x=-y$) the formulas match
continuously. This covers all possible $(x,y)\in\bbR^2$. In summary,
we find that $G(x,y)\equiv\cI(x,y)$ on $\bbR^2$. \qed

%=============================================================%
\section{The list of $24$ signed knots}
\label{app:lsit24pts}
%=============================================================%

For each $\pi\in S_4$, we have defined
\begin{eqnarray*}
\begin{cases}
u_\pi(\bdlambda)= \lambda_{\pi(1)}+\lambda_{\pi(2)}-\lambda_{\pi(3)}-\lambda_{\pi(4)} = 2(\lambda_{\pi(1)}+\lambda_{\pi(2)})-1,\\
v_\pi(\bdlambda)=
\lambda_{\pi(1)}-\lambda_{\pi(2)}+\lambda_{\pi(3)}-\lambda_{\pi(4)}
= 2(\lambda_{\pi(1)}+\lambda_{\pi(3)})-1.
\end{cases}
\end{eqnarray*}
Now we list all $24$ points $(u_\pi(\bdlambda),v_\pi(\bdlambda))$ in
the Table~\ref{tab:1}.
\begin{table}[h]
\centering \caption{The list of $24$ points}\label{tab:1}
\begin{tabular}{|c|c|c|c|}
\hline $\pi$ & $(u_\pi(\bdlambda),v_\pi(\bdlambda))$ & $\pi$ &
$(u_\pi(\bdlambda),v_\pi(\bdlambda))$\\
\hline (1) & $(2(\lambda_1+\lambda_2)-1,2(\lambda_1+\lambda_3)-1)$  &  (124) & $(2(\lambda_2+\lambda_4)-1,2(\lambda_2+\lambda_3)-1)$  \\
\hline (12) & $(2(\lambda_1+\lambda_2)-1,2(\lambda_2+\lambda_3)-1)$ &  (142) & $(2(\lambda_1+\lambda_4)-1,2(\lambda_3+\lambda_4)-1)$ \\
\hline(13) & $(2(\lambda_2+\lambda_3)-1,2(\lambda_1+\lambda_3)-1)$ & (134) & $(2(\lambda_2+\lambda_3)-1,2(\lambda_3+\lambda_4)-1)$  \\
\hline(14) & $(2(\lambda_2+\lambda_4)-1,2(\lambda_3+\lambda_4)-1)$ & (143) & $(2(\lambda_2+\lambda_4)-1,2(\lambda_1+\lambda_4)-1)$ \\
\hline(23) & $(2(\lambda_1+\lambda_3)-1,2(\lambda_1+\lambda_2)-1)$ & (234) & $(2(\lambda_1+\lambda_3)-1,2(\lambda_1+\lambda_4)-1)$\\
\hline(24) & $(2(\lambda_1+\lambda_4)-1,2(\lambda_1+\lambda_3)-1)$ & (243) & $(2(\lambda_1+\lambda_4)-1,2(\lambda_1+\lambda_2)-1)$\\
\hline (34) &  $(2(\lambda_1+\lambda_2)-1,2(\lambda_1+\lambda_4)-1)$ & (1234) & $(2(\lambda_2+\lambda_3)-1,2(\lambda_2+\lambda_4)-1)$\\
\hline(12)(34) & $(2(\lambda_1+\lambda_2)-1,2(\lambda_2+\lambda_4)-1)$ & (1243) & $(2(\lambda_2+\lambda_4)-1,2(\lambda_1+\lambda_2)-1)$\\
\hline(13)(24) & $(2(\lambda_3+\lambda_4)-1,2(\lambda_1+\lambda_3)-1)$ & (1324) & $(2(\lambda_3+\lambda_4)-1,2(\lambda_2+\lambda_3)-1)$ \\
\hline(14)(23) &  $(2(\lambda_3+\lambda_4)-1,2(\lambda_2+\lambda_4)-1)$ & (1342) & $(2(\lambda_1+\lambda_3)-1,2(\lambda_3+\lambda_4)-1)$ \\
\hline(123) & $(2(\lambda_2+\lambda_3)-1,2(\lambda_1+\lambda_2)-1)$ & (1423) & $(2(\lambda_3+\lambda_4)-1,2(\lambda_1+\lambda_4)-1)$ \\
\hline (132) &
$(2(\lambda_1+\lambda_3)-1,2(\lambda_2+\lambda_3)-1)$ & (1432) & $(2(\lambda_1+\lambda_4)-1,2(\lambda_2+\lambda_4)-1)$ \\
\hline
\end{tabular}
\end{table}
Let $x_{ij}:=2(\lambda_i+\lambda_j)-1$. Because
$\sum^4_{i=1}\lambda_i=1$, complementary pairs give opposite values:
\begin{eqnarray*}
x_{12}=-x_{34},\quad x_{13}=-x_{24},\quad x_{14}=-x_{23}.
\end{eqnarray*}
Define
\begin{eqnarray*}
\begin{cases}
\tilde\alpha&:=x_{12}=2(\lambda_1+\lambda_2)-1,\\
\tilde\beta&:=x_{13} =2(\lambda_1+\lambda_3)-1,\\
\tilde\gamma&:=x_{14}=2(\lambda_1+\lambda_4)-1.
\end{cases}
\end{eqnarray*}
It is easily seen that
$$
1>\tilde\alpha>\tilde\beta>\abs{\tilde\gamma}\geqslant0.
$$
For any $\pi\in S_4$, $u_\pi(\bdlambda)=x_{\pi(1)\pi(2)}$ and
$v_\pi(\bdlambda)=x_{\pi(1)\pi(3)}$. The two pairs
$\set{\pi(1),\pi(2)}$ and $\set{\pi(1),\pi(3)}$ share one index and
therefore cannot be equal or complementary. The six possible
pair-values are $\pm\tilde\alpha,\pm\tilde\beta,\pm\tilde\gamma$.
More precisely, the $24$ indexed points are of the forms
$$
\Set{(\pm\tilde\alpha,\pm\tilde\beta),(\pm\tilde\beta,\pm\tilde\alpha),(\pm\tilde\alpha,\pm\tilde\gamma),(\pm\tilde\gamma,\pm\tilde\alpha),(\pm\tilde\beta,\pm\tilde\gamma),(\pm\tilde\gamma,\pm\tilde\beta)}.
$$
Since $\abs{\tilde\gamma}<\tilde\beta<\tilde\alpha$, every such
point satisfies
$$
\abs{u}\leqslant \tilde\alpha,\abs{v}\leqslant\tilde\alpha,
\abs{u}+\abs{v}\leqslant \tilde\alpha+\tilde\beta.
$$
On the other hand, all eight points
$(\pm\tilde\alpha,\pm\tilde\beta)$ and
$(\pm\tilde\beta,\pm\tilde\alpha)$ occur among the listed points.
These are exactly the vertices of the region determined by the
preceding inequalities. For reference, in cyclic order they
correspond to:
\begin{eqnarray*}
\begin{array}{c|c}
\text{Vertex}&\text{One corresponding permutation}\\ \hline
(\tilde\alpha,\tilde\beta)&(1)\\
(\tilde\beta,\tilde\alpha)&(23)\\
(-\tilde\beta,\tilde\alpha)&(1243)\\
(-\tilde\alpha,\tilde\beta)&(13)(24)\\
(-\tilde\alpha,-\tilde\beta)&(14)(23)\\
(-\tilde\beta,-\tilde\alpha)&(14)\\
(\tilde\beta,-\tilde\alpha)&(1342)\\
(\tilde\alpha,-\tilde\beta)&(12)(34)
\end{array}
\end{eqnarray*}
Therefore, the terms involving $\tilde\gamma$ do not produce any
additional extreme points.

Besides, the essential point is that we need not only the $24$
points but also their signs. The twelve even permutations produce
$+$; and twelve odd permutations product $-$.
\begin{table}[h]
\centering \caption{The list of $24$ points with their
signs}\label{tab:2}
\begin{tabular}{|c|c|c|c|}
\hline $(u,v)$& $\sign(\pi)=\varepsilon(u,v)$ & $(u,v)$ & $\sign(\pi)=\varepsilon(u,v)$ \\
\hline
$(\tilde\alpha,\tilde\beta)$ & $+$ & $(\tilde\alpha,-\tilde\gamma)$ & $-$ \\
$(\tilde\alpha,-\tilde\beta)$ & $+$ & $(-\tilde\gamma,\tilde\beta)$ & $-$ \\
$(-\tilde\alpha,\tilde\beta)$ & $+$ & $(-\tilde\beta,-\tilde\alpha)$ & $-$ \\
$(-\tilde\alpha,-\tilde\beta)$ & $+$ & $(\tilde\beta,\tilde\alpha)$ & $-$ \\
$(-\tilde\gamma,\tilde\alpha)$ & $+$ & $(\tilde\gamma,\tilde\beta)$ & $-$ \\
$(\tilde\beta,-\tilde\gamma)$ & $+$ & $(\tilde\alpha,\tilde\gamma)$ & $-$ \\
$(-\tilde\beta,-\tilde\gamma)$ & $+$ & $(-\tilde\gamma,-\tilde\beta)$ & $-$ \\
$(\tilde\gamma,-\tilde\alpha)$ & $+$ & $(-\tilde\beta,\tilde\alpha)$ & $-$ \\
$(-\tilde\gamma,-\tilde\alpha)$ & $+$ & $(-\tilde\alpha,-\tilde\gamma)$ & $-$ \\
$(-\tilde\beta,\tilde\gamma)$ & $+$ & $(\tilde\beta,-\tilde\alpha)$ & $-$ \\
$(\tilde\beta,\tilde\gamma)$ & $+$ & $(-\tilde\alpha,\tilde\gamma)$ & $-$ \\
$(\tilde\gamma,\tilde\alpha)$ & $+$ & $(\tilde\gamma,-\tilde\beta)$ & $-$ \\
\hline
\end{tabular}
\end{table}
We will write $(u,v)^\pm$ for one point $(u,v)$ from $24$ signed
knots with corresponding $+$ or $-$ that is identified from the
above table. For instance, $(u,v)^+$ means
$(u,v)=(u_\pi(\bdlambda),v_\pi(\bdlambda))$ for $\pi\in S_4$ with
$\varepsilon(u,v)=\varepsilon(u_\pi(\bdlambda),v_\pi(\bdlambda))=\sign(\pi)=+1$.

%=====================================================================%
\section{The computational details of
Eq.~\eqref{eq:Tintegral}}\label{app:Txy}
%=====================================================================%

One antiderivative is
\begin{eqnarray*}
F_{x,y}(r) = \frac{(xy)^2}3 r^3 + \frac{xy(x-y)}2 r^4 +
\frac{x^2-4xy+y^2}5 r^5 + \frac{y-x}3 r^6 + \frac17 r^7.
\end{eqnarray*}
Thus, for $x\geqslant \max(0,-y)$, it holds that
\begin{eqnarray*}
\cT(x,y)=\frac18\Br{F_{x,y}(x) - F_{x,y}(\max(0,-y))}.
\end{eqnarray*}
The specific details of calculation is presented below:
\begin{itemize}
\item \textbf{Case 1: $y\geqslant0$.} At this time, $F_{x,y}(x)=\frac{x^5y^2}{30} + \frac{x^6y}{30} + \frac{x^7}{105} $, thus
$$
\cT(x,y)=\frac{x^5(2x^2+7xy+7y^2)}{1680}.
$$
\item  \textbf{Case 2: $-x\leqslant y<0$.}  Let $z=x+y\geqslant0$. Perform the transformation $u=r+y$, we get that
\begin{eqnarray*}
\cT(x,y)&=&\frac18\int^{x+y}_0(x+y-u)^2u^2(u-y)^2\dif u\\
&=&\frac{(x+y)^5(2x^2-3xy+2y^2)}{1680}.
\end{eqnarray*}
\item  \textbf{Case 3: $x<0$ or $x<-y$.}  At this time, the interval of integration is empty set, thus $\cT(x,y)=0$.
\end{itemize}
In summary, we get the desired conclusion. \qed

%--------------------------------------------------%

\end{document}

%% file: qudit_preamble.tex
\newcommand{\cmp}{Comm. Math. Phys.~}

\newcommand{\jmp}{J. Math. Phys.~}
\newcommand{\jfa}{J. Funct. Anal.~}
\newcommand{\jpa}{J. Phys. A~}

\usepackage{young}
\usepackage[latin1]{inputenc}
\usepackage{appendix}
\usepackage{epstopdf}
\usepackage{xcolor}
\usepackage{subfigure}
\usepackage[T1]{fontenc}
\usepackage[sc]{mathpazo}
\usepackage{amsmath}
\usepackage{amssymb}
\usepackage{graphicx,color}
\usepackage{framed}
\usepackage{multirow}
\usepackage{enumerate}
\usepackage{amsthm}% theorems and proofs
\usepackage{amsfonts,mathrsfs}
\usepackage{geometry} % allows easy specifying of the page layout
\usepackage{eepic}
\usepackage{ifthen}
\usepackage[vcentermath]{youngtab}
\usepackage[unicode=true,pdfusetitle, bookmarks=true,bookmarksnumbered=false,bookmarksopen=false, breaklinks=false,pdfborder={0 0 0},backref=false,colorlinks=false] {hyperref}
\hypersetup{
colorlinks,linkcolor=myurlcolor,citecolor=myurlcolor,urlcolor=myurlcolor}
\definecolor{myurlcolor}{rgb}{0,0,0.7}% Color highlighted for numbers in citations and theorems

\usepackage{color}

\newcommand{\blue}{\textcolor{blue}}

\newcommand{\proj}[1]{| #1\rangle\!\langle #1 |}

\usepackage{hyperref}
\hypersetup{pdfpagemode=UseNone}

\newcommand{\tinyspace}{\mspace{1mu}}

\newcommand{\op}[1]{\operatorname{#1}}

\newcommand{\abs}[1]{\left\lvert\tinyspace #1 \tinyspace\right\rvert}

\newcommand{\norm}[1]{\left\lVert\tinyspace #1 \tinyspace\right\rVert}

\renewcommand{\det}{\operatorname{det}}
\renewcommand{\t}{{\scriptscriptstyle\mathsf{T}}}

\newcommand{\sign}{\op{sign}}

\def\SO{\mathsf{SO}}
\def\SU{\mathsf{SU}}

\def\vol{\mathrm{vol}}

\def\haar{\mathrm{Haar}}

\def \dif {\mathrm{d}}
\def \diag {\mathrm{diag}}
\def \vol {\mathrm{vol}}

\def\I{\mathbb{1}}

\def\zero{\mathbf{0}}

\def\bdlambda{\boldsymbol{\lambda}}

\def\bdbeta{\boldsymbol{\beta}}
\def\bdomega{\boldsymbol{\omega}}
\def\bdsigma{\boldsymbol{\sigma}}

\newenvironment{mylist}[1]{\begin{list}{}{
    \setlength{\leftmargin}{#1}
    \setlength{\rightmargin}{0mm}
    \setlength{\labelsep}{2mm}
    \setlength{\labelwidth}{8mm}
    \setlength{\itemsep}{0mm}}}
    {\end{list}}

\def\ot{\otimes}

\newcommand{\Inner}[2]{\left\langle #1 , #2\right\rangle}
\newcommand{\Innerm}[3]{\left\langle #1 \left| #2 \right| #3 \right\rangle}

\newcommand{\Herm}{\mathrm{Herm}}

\newcommand{\Spec}{\mathrm{Spec}}

\newcommand{\pa}[1]{(#1)}
\newcommand{\Pa}[1]{\left(#1\right)}

\newcommand{\Br}[1]{\left[#1\right]}
\newcommand{\set}[1]{\{#1\}}
\newcommand{\Set}[1]{\left\{#1\right\}}

\newcommand{\ket}[1]{|#1\rangle}

\DeclareMathOperator{\trace}{Tr}
\newcommand{\ptr}[2]{\trace_{#1}\pa{#2}}
\newcommand{\Ptr}[2]{\trace_{#1}\Pa{#2}}

\newcommand{\Tr}[1]{\Ptr{}{#1}}

\newcommand{\Abs}[1]{\left|\tinyspace#1\tinyspace\right|}

\def\cF{\mathcal{F}}\def\cH{\mathcal{H}}\def\cI{\mathcal{I}}
\def\cK{\mathcal{K}}\def\cL{\mathcal{L}}
\def\cT{\mathcal{T}}
\def\cU{\mathcal{U}}\def\cW{\mathcal{W}}

\def\bsA{\boldsymbol{A}}\def\bsB{\boldsymbol{B}}
\def\bsH{\boldsymbol{H}}

\def\bsP{\boldsymbol{P}}\def\bsR{\boldsymbol{R}}
\def\bsU{\boldsymbol{U}}\def\bsV{\boldsymbol{V}}\def\bsW{\boldsymbol{W}}\def\bsX{\boldsymbol{X}}\def\bsY{\boldsymbol{Y}}

\def\bsa{\boldsymbol{a}}\def\bsb{\boldsymbol{b}}
\def\bsh{\boldsymbol{h}}

\def\bsr{\boldsymbol{r}}\def\bss{\boldsymbol{s}}\def\bst{\boldsymbol{t}}
\def\bsw{\boldsymbol{w}}\def\bsx{\boldsymbol{x}}\def\bsy{\boldsymbol{y}}
\def\bsz{\boldsymbol{z}}

\def\rD{\mathrm{D}}

\def\bbC{\mathbb{C}}\def\bbE{\mathbb{E}}

\def\bbR{\mathbb{R}}

\def\bbZ{\mathbb{Z}}

\def\sfU{\mathsf{U}}

\newtheorem{thrm}{Theorem}[section]
\newtheorem{lem}[thrm]{Lemma}
\newtheorem{prop}[thrm]{Proposition}
\newtheorem{cor}[thrm]{Corollary}
\theoremstyle{definition}
\newtheorem{definition}[thrm]{Definition}
\newtheorem{remark}[thrm]{Remark}

\numberwithin{equation}{section}

\newcounter{questionnumber}

%% file: Marginal_spectral_distributions.bbl
\begin{thebibliography}{999}


\bibitem{Boor1993}
C. de Boor, K. H\"{o}llig, S. Riemenschneider, Box Splines,
Springer-Verlag New York, Inc (1993).


\bibitem{Bravyi2004}
S. Bravyi, Requirements for compatibility between local and
multipartite quantum states, Quantum Information \& Computation,
\href{https://dl.acm.org/doi/10.5555/2011572.2011574}{{\bf4}, 12-26
(2004).}

\bibitem{Bytsenko2005}
A.A. Bytsenko, M. Libine, and F.L. Williams, Localization of
equivariant cohomology for compact and non-compact group actions,
\href{https://doi.org/10.1080/1726037X.2005.10698497}{{\bf3},
171-195 (2005).}

\bibitem{Christandl2014}
M. Christandl, B. Doran, S. Kousidis, and M. Walter,  Eigenvalue
distributions of reduced density matrices, \cmp
\href{https://doi.org/10.1007/s00220-014-2144-4}{{\bf332}, 1-52
(2014).}

\bibitem{Collins2023}
B. Collins and C. McSwiggen, Projections of orbital measures and
quantum marginal problems, Trans. Amer. Math. Soc.
\href{https://doi.org/10.1090/tran/8931}{{\bf376}, 5601-5640
(2023).}


\bibitem{Davidson1976}
E.R. Davidson, Reduced density matrices in quantum chemistry,
Academic Press (1976).

\bibitem{Duistermaat2010}
J.J. Duistermaat and  J.A.C. Kolk, Distributions: Theory and
Applications,  Birkh\"{a}user Boston, Boston, MA, (2010).

\bibitem{DH1982}
J.J. Duistermaat and G.J. Heckman, On the variation in the
cohomology of the symplectic form of the reduced phase space,
Invent. Math. \href{https://doi.org/10.1007/BF01399506}{{\bf69},
259-268 (1982).}

\bibitem{Hall2015}
B.C. Hall, Lie Groups, Lie Algebras, and Representations, Springer
International Publishing Switzerland (2015).

\bibitem{Harish1975}
Harish-Chandra, Harmonic analysis on real reductive groups I: the
theory of the constant term, \jfa
\href{https://doi.org/10.1016/0022-1236(75)90034-8}{{\bf19}, 104-204
(1975).}

\bibitem{IZ1980}
C. Itzykson and J.B. Zuber, The planar approximation. II. \jmp
\href{https://doi.org/10.1063/1.524438}{{\bf21},411 (1980).}

\bibitem{Kirwan1984}
F.C. Kirwan, Cohomology of Quotients in Symplectic and Algebraic
Geometry, Princeton University Press (1984).

\bibitem{Klyachko2004}
A.A. Klyachko, Quantum marginal problem and representations of the
symmetric group,
\href{https://arxiv.org/abs/quant-ph/0409113}{arXiv:quant-ph/0409113}


\bibitem{McSwiggen2018}
C. McSwiggen, A new proof of Harish-Chandra's integral formula, \cmp
\href{https://doi.org/10.1007/s00220-018-3259-9}{{\bf365}, 239-253
(2018).}

\bibitem{Mejia2017}
J. Mej\'{i}a, C. Zapata, A. Botero, The difference between two
random mixed quantum states: exact and asymptotic spectral analysis,
\jpa: Math. Theor.
\href{https://doi.org/10.1088/1751-8121/50/2/025301}{{\bf50}, 025301
(2017).}

\bibitem{Olshanski2013}
G. Olshanski, Projections of orbital measures, Gelfand-Tsetlin
polytopes, and splines, Journal of Lie Theory,
\href{https://www.heldermann.de/JLT/JLT23/JLT234/jlt23050.htm}{{\bf23},
1011-1022 (2013).}

\bibitem{Silva2008}
A.C. Silva, Lectures on Symplectic Geometry, Springer-Verlag (2008).

\bibitem{Walter2014}
M. Walter, Multipartite quantum states and their marginals, PhD
Thesis \href{https://arxiv.org/abs/1410.6820}{arXiv:1410.6820}

\bibitem{Wang1994}
R. Wang, Multivariate Spline Functions and Their Applications,
Science Press, Beijing, P. R. China (1994).

\bibitem{Zhang2017}
L. Zhang, Average coherence and its typicality for random mixed
quantum states, \jpa: Math. Theor.
\href{https://doi.org/10.1088/1751-8121/aa6179}{{\bf50}, 155303
(2017).}

\bibitem{Zhang2021}
L. Zhang, Dirac delta function of matrix argument, Int. J. Theor.
Phys. \href{https://doi.org/10.1007/s10773-020-04598-8}{{\bf60},
2445-2472 (2021).}

\bibitem{Zhang2019}
L. Zhang, Y. Jiang, and J. Wu, Duistermaat-Heckman measure and the
mixture of quantum states, \jpa: Math. Theor.
\href{https://doi.org/10.1088/1751-8121/ab5297}{{\bf52}, 495203
(2019).}

\end{thebibliography}
